\documentclass[a4paper,12pt]{article}

\usepackage{setspace}
\usepackage{stmaryrd}
\usepackage{amsmath}
\usepackage{amssymb}
\usepackage{comment}
\usepackage{mathrsfs}
\usepackage{graphicx}
\usepackage{graphics,natbib,anysize}
\usepackage[colorlinks, citecolor=blue, urlcolor=blue, pdfborder={0 0 1}, citebordercolor={0 1 0}]{hyperref}
\usepackage{tikz}

\usetikzlibrary{arrows}
\newenvironment{proof}[1][Proof]{\begin{trivlist}
\item[\hskip \labelsep {\bfseries #1}]}{\end{trivlist}}
\newtheorem{prop}{{Proposition}}

\title{\large \textsc{\bf{Communication Learning in Social Networks: \\Finite Population and the Rates}\footnote{We thank Daron Acemoglu, S$\acute{e}$bastien Bubeck, John Campbell, Emmanuel Farhi, Drew Fudenberg, Matthew Jackson, Gareth James, Philip Reny, Philippe Rigollet, Andrei Shleifer, Alp Simsek, and Yiqing Xing for valuable comments and suggestions. All errors are ours.}}}
\author{\normalsize \textbf{Jianqing Fan}\footnote{ \sloppy Department of Operations Research and Financial Engineering, Princeton University.} \\ \normalsize Princeton University \and \normalsize \textbf{Xin Tong}\footnote{ \sloppy Department of Mathematics, Massachusetts Institute of Technology.} \\ \normalsize MIT \and \normalsize \textbf{Yao Zeng}\footnote{ \sloppy Department of Economics, Harvard University.}\\ \normalsize Harvard University}

\date{}
\begin{document}
\maketitle
%\vspace*{0.5cm}

\begin{abstract}
Following the Bayesian communication learning paradigm, we propose a \textit{finite population learning} concept to capture the level of information aggregation in any given network, where agents are allowed to communicate with neighbors repeatedly before making a single decision.  This concept helps determine the occurrence of effective information aggregation in a finite network and reveals explicit interplays among parameters.  It also enables meaningful comparative statics regarding the effectiveness of information aggregation in networks.  Moreover, it offers a solid foundation to address, with a new perfect learning concept, long run dynamics of learning behavior and the associated learning rates as population diverges.   Our conditions for the occurrence of finite population learning and perfect learning in communication networks are very tractable and transparent.
\end{abstract}
\vspace*{0.2cm}

\small \textsc{Keywords:} Social networks, Bayesian update, communication, finite population learning, perfect learning, learning rates.
%%
%%\vspace*{0.2cm}
%%\textcolor{red}{Check Qiu's paper: information acquisition? Other reference.}
%%\end{abstract}
%%\vspace*{1.7cm}
%%\thispagestyle{empty}
%
%
%
%\newtheorem{prop}{\textsc{\sc{Proposition}}}

%\newpage
%\pagenumbering{arabic} \setcounter{page}{1}

\newtheorem{ass}{Assumption}
\newtheorem{theo}{Theorem}
\newtheorem{lem}{\textsc{\sc{Lemma}}}
\newtheorem{defi}{Definition}
\newtheorem{coro}{\textsc{\sc{Corollary}}}
\newtheorem{remark}{\textsc{\sc{Remark}}}
\newtheorem{examp}{Example}

\newcommand{\info}{\mathrm{info}}
\newcommand{\erf}{\mathrm{erf}}
\newcommand{\SI}{\mathrm{SI}}
\newcommand{\EI}{\mathrm{EI}}
\newpage
\normalsize
\section{Introduction}

The effectiveness of information aggregation has been long and widely recognized as a central theme for good decision making at both individual and aggregate levels.  Boosted by the Internet and particularly online social networks, this theme is especially important in communication and decision making in the modern world. People communicate with their friends, through extremely efficient, open and multi-dimensional approaches, in social networks before making specific decisions.  In particular, these circumstances of information exchange often involve strategic interactions among people, which call for new modeling techniques beyond mainstream statistics and economics literature, such as the convolution of game theory and graphical models.

Recently, \cite{ABO2011} provide a fascinating model to study communication in social networks and the implications for information aggregation.  They employ a game-theoretic framework to model people's information aggregation in social networks.  They define an intuitive concept of asymptotic learning, which means as the population of a network diverges, the probability that a large fraction of people take ``correct" actions converges to one or eventually exceeds a high threshold. Given agents communicate either truthfully or strategically, they establish equilibrium conditions under which asymptotic learning occurs. They also discuss the welfare implications of asymptotic learning, and investigate the impacts of specific types of cost structures and social cliques.

Motivated by the asymptotic learning concept, we ask the following questions. Can we define a good communication learning concept regarding a finite population network?  If so, does such learning occur in a given finite social network? What are necessary and sufficient conditions to guarantee such learning? Can we write down clean and tractable rates at which a society achieves long run asymptotic learning?
These questions are relevant and important, because it is common practice for people to assess the effectiveness of information aggregation  in given organizations, regions or nations.
Such assessment regarding finite population networks naturally offers a solid foundation for people to understand the quality of social learning when the society evolves.  As of now, current researchers in social networks have not provided desirable answers to these questions, and previous works called for fresh inputs (\cite{G2009}; \cite{AO2010}; \cite{J2010}).

Based on an information exchange game in social networks modified from \cite{ABO2011}, we propose a \textit{finite population learning} concept, which captures the level of aggregation of disperse information in any given communication network.  In the model, there is an underlying state.  People in a social network do not know the underlying state, but they have a common prior on the distribution of the state. After receiving initial private signals related to the underlying state, they exchange information simultaneously in the network, at times specified by a homogenous Poisson process,  until taking an irreversible action to exit the network. Upon each person's exit, she makes an estimate of the underlying state.
Her payoff depends on the waiting time before making the decision and the expected mean-square error between her estimate and the underlying state.  The longer she waits, the more information she gathers and hence the better her estimate is, but the more discounting incurs. Thus, she needs to take a prompt action after obtaining sufficient amount of information in the network.

The newly defined finite population learning concept involves three parameters, $\epsilon$, $\bar{\epsilon}$, and $\delta$ for a given social network $G_n$ of population size $n$; rigorously, it is called $(\epsilon, \bar\epsilon, \delta)$-learning.  The parameter $\epsilon$ is the precision under which an agent's decision is considered ``correct", $1-\bar\epsilon$ represents the fraction of agents in the network who make the approximately correct decision, and $1-\delta$ represents the probability at which such a fraction of agents make the approximately correct decision.  We think of these three parameters as tolerance parameters of finite population learning.
To contrast with asymptotically driven concepts,  $(\epsilon, \bar\epsilon, \delta)$-learning is simply referred to as finite population learning in verbal discussions.

We derive necessary and sufficient conditions for the occurrence of finite population learning under any given equilibrium.  Intuitively, finite population learning is more likely to occur when the number of signals an agent obtains under equilibrium is larger, or the tolerances of learning are larger.
Interestingly, the impact of the information precisions on finite population learning is ambiguous, which parallels the well-known Hirshleifer effect and subsequent work on the social value of information but stems from a new and different mechanism.  We also provide necessary and sufficient conditions for the occurrence of finite population learning under \textit{any equilibrium}, namely, without knowledge of a particular equilibrium.

%Our conditions are easy to check, and they are helpful in the sense that they exhibit explicit interplays among parameters, such as tolerances, information precisions and information-sensitiveness of the decision problem.

A straightforward advantage of our conditions is that these conditions lead to meaningful comparative statics regarding the effectiveness of information aggregation in networks.  In these conditions, the underlying forces, such as tolerances, information precisions and information-sensitiveness, that shape the effectiveness of information aggregation in a given finite communication network are explicitly displayed in a single formula.  Compared to the asymptotic learning results in previous literature, our conditions for finite population learning involve only one equilibrium outcome, which is the number of signals an agent obtains when she exits under equilibrium.  More importantly, different from our finite population learning concept in which the total amount of information is fixed, the existing asymptotic learning literature employs an implicit assumption that the total amount of information grows linearly with the population size.  Hence, the learning status with respect to a sequence of networks with growing population reflects not only the effectiveness of information aggregation of certain network structures, but also an increased endowment of total information. Our finite population learning concept overcomes this defect and disentangles the effectiveness of information aggregation from the growth of information endowment.

%The tractability and transparency of our conditions for finite population learning depends on our observation that the equilibrium action of every agent at any time $t$ depends on $t$ only through the dependence on the number of communication steps she experienced up to the time $t$.  We highlight this property as \textit{$|T_t|$-dependence}. The $|T_t|$-dependence enables us to construct a finite game with complete information as a reduced-form of the initial information exchange game.  This reduced-form game is equivalent to its corresponding initial information exchange game in the sense that they lead to the same consequence of information aggregation and learning.  Specifically, the numbers of signals an agent obtains when she exits under equilibrium are the same in the two games with consistent parameters.
%What makes the identification more powerful is that the equilibrium of the reduced-form game is easy to compute numerically. This further enhances the values of our conditions for finite population learning and of the concept of finite population learning itself.  Also, by analyzing reduced-form games, we can relax the dependency of finite population learning on a particular equilibrium, by investigating social learning under \textit{all equilibria} of a communication network without equilibrium selection.%, namely, with knowledge of multiple equilibria but not knowing which equilibrium actually occurs.

The finite population learning concept enables us to investigate the rate at which a sequence of growing communication networks $\{G_n\}_{n=1}^\infty$, which is referred to as a \emph{society}, reaches \emph{perfect learning}.  Perfect learning occurs if all communication networks in a society achieve finite population learning under vanishing tolerances as population grows. For example, 
we say $\delta$-perfect learning occurs along society $\{G_n\}_{n=1}^\infty$ if i). $(\epsilon, \bar\epsilon, \delta_n)$-learning occurs for each network $G_n$ in the society, and ii). $\delta_n$ goes to zero as $n$ goes to infinity. The learning rate is characterized by the sequence $\{\delta_n\}_{n=1}^\infty$.   Clearly, faster learning rate implies perfect learning is reached at a higher quality.
It is instructive to distinguish our learning rate concept from the speed of convergence to a pre-defined consensus in existing social learning literature, which mainly concerns about the time towards a consensus in a circumstance where people make repeated decisions and learn from others' previous decisions to help to make their own future decisions.  In such a context, the observable sequence of aggregate decisions naturally reveals the dynamics of information aggregation along the time dimension.  In our story of direct communication, however, although people communicate with each other repeatedly, they only make a single decision, and different people may go through varying communication rounds before their decisions.  This makes the time dynamics of information aggregation largely unobservable, and thus calls for alternative dimensions to look into the information dynamics.  %As will be analyzed in depth, our concept of learning rate gears towards the nature of direct communication, and helps investigate the speed at which the long run distribution of social learning is reached in different societies with different patterns of population growth.

We have given conditions for societies to reach $\delta$-perfect learning at a certain desired rate $\{\delta_n\}_{n=1}^\infty$.
Given a sequence of networks and the associated equilibria, we define an \textit{equilibrium informed agent} as one who obtains an unbounded number of signals as the population size goes to infinity.  The $\delta$-perfect learning occurs if almost all agents in the society are equilibrium informed.
Moreover, without involving any equilibrium, we define a \textit{socially informed agent} (roughly) as one who has an unbounded number of neighbors in a finite distance as the population goes to infinity.  The $\delta$-perfect learning occurs if almost all agents are socially informed.
%Compared to the perfect learning result in \cite{ABO2011}, our sufficient conditions for perfect learning are more transparent and easier to interpret.
We also explicitly explore the achievable fastest learning rate for perfect learning in a given society.  Under some circumstances, achievable learning rate could be in the exponential order.  This implies that a society with growing population might achieve a desirable level of finite population learning very quickly.\\  %Our new learning results lead to sharper social and economic implications and offer new insights on the associated information exchange in social networks.

\textsc{Relation to Literature.} Our work lies in the category of Bayesian social learning in social networks, in which decision makers in a social network update their information according to the Bayes' rule.  %Agents obtain new information from their neighbors in the network, and they this new information to update their prior to get posterior following the Bayes' rule.
General Bayesian social learning is divided into two sub-categories, namely Bayesian observational learning and Bayesian communication learning. In Bayesian observational learning, agents observe past actions of their neighbors.  From these observed actions, agents update their beliefs and make inferences.  Herd behavior is a very typical consequence of observational learning. %, which describes how individuals in a group can act together without planned direction.
In literature, \cite{B1992}, \cite{BHW1992} and \cite{SS2000} are early attempts to model herd effects through Bayesian observational learning.  \cite{BF2004} and \cite{SS2008} relax the assumption of full observation network topology and study Bayesian observational learning with sampling of past actions.  Recently, \cite{ADLO2011} and \cite{M2011} investigate how detailed network structures could add new interesting insights.

Our work belongs to Bayesian communication learning, which means that agents cannot directly observe actions of others but can communicate with each other before making a decision.  Consequently, agents update their beliefs and make inferences based on the information given by others.  New interesting considerations arise in Bayesian communication learning; for example, agents may not want to truthfully reveal their information to others through communication. %and even if communication is truthful, at one time agents may also only get partial information from their neighbors, which does not necessarily determine the final decisions.
\cite{CS1982} pioneers the research in strategic communication, and \cite{ABO2011} is an interesting piece that looks into how communication learning shapes information aggregation in social networks.  Other works such as \cite{GGS2010} and \cite{HK2010} also study strategic communication in social networks, but their focus is not on information aggregation.

There is a branch of literature that applies various non-Bayesian updating methods to investigate information aggregation and social learning. \cite{D1974} develops a tractable non-Bayesian learning model which is frequently employed in research of social networks today. Essentially, the DeGroot model is pertaining to observational learning, in which agents make today's decisions by taking the average of neighbors' beliefs revealed in their decisions yesterday.  \cite{DVZ2003} and \cite{GJ2010, GJ2012b, GJ2012a,  GJ2011} apply the DeGroot model to financial networks and general social networks, respectively.  By a field experiment, \cite{MPS2010} compares a non-Bayesian model of communication with a model in which agents communicate their signals and update information based on Bayes' rule.  Their evidence is generally in favor of the Bayesian communication learning approach.  %We do not discuss non-Bayesian updating in our paper.

Our paper is most related to \cite{ABO2011}.  Compared to their work, we employ a simplified framework for network communication and exploit more undeveloped mechanisms. In particular, we mainly focus on the effect of social learning and information aggregation in finite population communication networks. This allows for clear comparative statics with respect to learning, and for discussion on  the rates of learning as the population increases.  As of now, researchers have not provided desirable results in finite population communication network as well as results regarding learning rates. To the best of our knowledge, our work is the first attempt to address these questions with clear answers.

Our work is also related to \cite{GJ2012b, GJ2012a, GJ2011}, in particular on the investigation of learning rate. \cite{GJ2012b, GJ2012a, GJ2011} employ the DeGroot model to analyze the impacts of homophily in social networks, which refers to the tendency of agents to associate relatively more with those who are similar to them, on the learning rate in the context of observational learning.  Our results of learning rate are different from theirs in two aspects.  First, our focus is on Bayesian communication learning rather than non-Bayesian observational learning.  Second, as discussed before, our concept of learning rate is based on perfect learning as the population in networks diverges, rather than the time towards a consensus in their model.  An appealing feature of \cite{GJ2012b, GJ2012a, GJ2011} is that their results of learning rate are based on certain statistics of networks rather than the full network structures, which could lead to potentially more empirical traction.

We would also like to relate this work to social network papers in existing statistics literature.  The larger part of those papers are based on graphical models, which are ideal to describe structural formation.  Rather than providing a list of state-of-the-art contributions, we refer interested readers to   \cite{Newman2010} and \cite{Kolacazyk2009}, which might serve as a broad introduction to the field.   Our work supplements structural modeling with human behavior modeling through game theory. Such model enrichment is necessary for some specific objectives; for example, we will see that strategic interaction and contextual information that sit outside graphical models are crucial to determine the final information aggregation status.   Also, our theoretical results are in the same spirit of the finite sample results in the statistical learning theory, such as the Vapnik-Chervonenkis inequality.  Such results with a clear characterization of strategic interactions may have potential to expand the scope of the finite sample approach beyond statistical learning theory.

The rest of the paper is organized as follows.  Section 2 introduces the information exchange game and characterizes its equilibrium.  New finite population learning concept is proposed in Section 3.  Section 4 discusses dynamics of learning and addresses learning rates explicitly. In the final section, we discuss possible directions for further research.  All proofs are in the supplementary materials.

\section{The Model}
\setcounter{equation}{0}

%\textcolor{red}{Lao Lin: Probably we need to say why we employ such a game-theoretical model rather than standard graphic model in statistics. In other words, why do you, a statstician, deviate?}

In this section, we present our model of information exchange in social networks, which is closely related to \cite{ABO2011}, but has different focus.  In this model, people, formally called as agents, are organized in some network structure. Each agent has her initial information. Agents are able to solicit information from their neighbors through communication, restricted by the network structure and a communication clock. The communication clock defines the times at which each agent is able to communicate with others.  At each round of communication, agents are obliged to transmit truthfully \emph{all} information they have to their neighbors in the network.  By such communication, the information set of an agent can become larger as time evolves.  With the help of her initial and acquired information, every agent is able to make a decision and exit.  An exit strategy is needed due to the time value of information content.  After exit, an agent would not have any incentive to further acquire information from neighbors, but she is still obliged to transmit all her information to others in the next round of communication.  Through certain measure of agents' decisions, we are further able to characterize the quality of learning and information aggregation.

We make the following assumptions to simplify the analysis and focus on a concept of finite population learning, which will be rigorously defined in the next section.  First, we assume mandatory communication, which means that no agent holds her information to herself.  When communication times arrive, an agent has to send all her information set to all of her direct neighbors.  Second, we assume truthful communication, which means whenever an agent sends information, she has to send unmanipulated information, whether it is her own private information or obtained information originated from other agents.  
%Literally, these two assumptions are equivalent to assumption of truthful communication in \cite{ABO2011}, but as our focus in this paper is not to uncover similarity and differences between truthful and strategic communication, we do not plan to relax this assumption and make the model more general.  
We will first analyze communication and information aggregation in a given finite population network, and then consider the limit as the population grows to infinity.  In this course, we assume that existing links are kept when a network grows.

Before formal definition of  the game, we would like to illustrate how information flows with an example.  For simplicity, suppose there are four agents in the network below. At time $t=0$, each agent $i$ has some private signal $s_i$, which captures her initial information. So the total information endowment in the system is $\{s_1, s_2, s_3, s_4\}$.   Communication occurs at $t=1, 2$. Due to the structure of the graph in our example, there is no need to consider beyond the second  communication round, since no additional information will be communicated further.

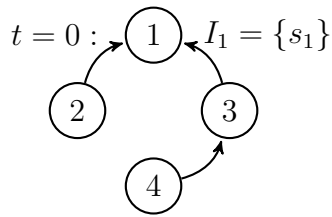
\begin{center}
\begin{tikzpicture}[->,>=stealth',shorten >=1pt,auto,node distance=2cm,
  thick,main node/.style={circle,fill=white!15,draw,font=\large\bfseries}]

  \node (1) at (0,0) [shape=circle, draw] {1} node at (1.5,0) {$I_1=\{s_1\}$} node at (-1.3,0) {$t=0:$};
  \node (2) at (-1,-1) [shape=circle, draw] {2};
  \node (3) at (1,-1) [shape=circle, draw] {3};
  \node (4) at (0,-2) [shape=circle, draw] {4};

  \path[every node/.style={font=\small}]
    (2) edge [bend left] (1)
    (3) edge [bend right] (1)
    (4) edge [bend right] (3);
\end{tikzpicture}
\end{center}

We will study two cases, and focus on agent $1$'s information set $I_1$. In the first case, suppose no agent exits after time $t=0$.  So the information flow is as follows:

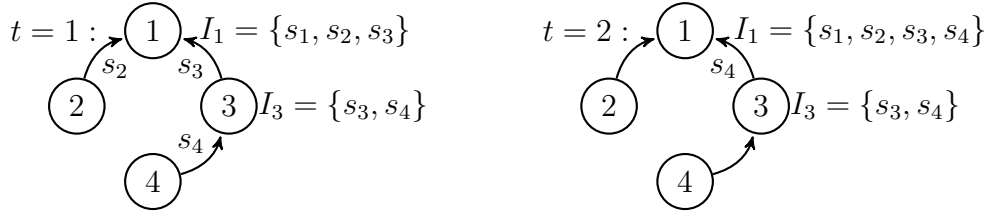
\begin{center}
\begin{tikzpicture}[->,>=stealth',shorten >=1pt,auto,node distance=2.5cm,
  thick,main node/.style={circle,fill=white!15,draw,font=\large\bfseries}]

\node (1) at (0,0) [shape=circle, draw] {1} node at (2,0) {$I_1=\{s_1, s_2, s_3\}$} node at (-1.3,0) {$t=1:$};
  \node (2) at (-1,-1) [shape=circle, draw] {2};
  \node (3) at (1,-1) [shape=circle, draw] {3} node at (2.5, -1) {$I_3=\{s_3, s_4\}$};
  \node (4) at (0,-2) [shape=circle, draw] {4};

  \path[every node/.style={font=\small}]
    (2) edge [bend left] (1)
    (3) edge [bend right] (1)
    (4) edge [bend right] (3);

  \node (i2) at (-0.5,-0.5)  {$s_2$};
  \node (i3) at (0.5, -0.5) {$s_3$};
  \node (i4) at (0.5, -1.5) {$s_4$};

\node (5) at (7,0) [shape=circle, draw] {1} node at (9.3,0) {$I_1=\{s_1, s_2, s_3, s_4\}$} node at (5.7,0) {$t=2:$};
  \node (6) at (6,-1) [shape=circle, draw] {2};
  \node (7) at (8,-1) [shape=circle, draw] {3} node at (9.5, -1) {$I_3=\{s_3, s_4\}$};
  \node (8) at (7,-2) [shape=circle, draw] {4};

  \path[every node/.style={font=\small}]
    (6) edge [bend left] (5)
    (7) edge [bend right] (5)
    (8) edge [bend right] (7);

  \node (i5) at (7.5, -0.5) {$s_4$};
\end{tikzpicture}
\end{center}

After the first round of communication, i.e., $t=1$, agent $1$ has signals $\{s_1, s_2, s_3\}$.  Also note that at this time agent $3$ has $\{s_3, s_4\}$. At $t=2$, agent $3$ sends the newly grabbed signal $s_4$ to agent $1$. So agent $1$'s information set enriches to $\{s_1, s_2, s_3, s_4\}$.

In the second case, suppose agent $3$ exits after time $t=0$, then she is still obliged to send all her signals (in this case, only her private signal) she acquires to neighbors, but she does not have any incentive to receive others' signal.  Therefore, the information flow is as follows.

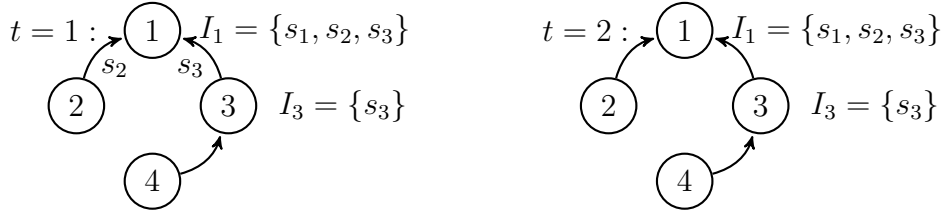
\begin{center}
\begin{tikzpicture}[->,>=stealth',shorten >=1pt,auto,node distance=2.5cm,
  thick,main node/.style={circle,fill=white!15,draw,font=\large\bfseries}]

\node (1) at (0,0) [shape=circle, draw] {1} node at (2,0) {$I_1=\{s_1, s_2, s_3\}$} node at (-1.3,0) {$t=1:$};
  \node (2) at (-1,-1) [shape=circle, draw] {2};
  \node (3) at (1,-1) [shape=circle, draw] {3} node at (2.5, -1) {$I_3=\{s_3\}$};
  \node (4) at (0,-2) [shape=circle, draw] {4};

  \path[every node/.style={font=\small}]
    (2) edge [bend left] (1)
    (3) edge [bend right] (1)
    (4) edge [bend right] (3);

  \node (i2) at (-0.5,-0.5)  {$s_2$};
  \node (i3) at (0.5, -0.5) {$s_3$};

\node (5) at (7,0) [shape=circle, draw] {1} node at (9,0) {$I_1=\{s_1, s_2, s_3\}$} node at (5.7,0) {$t=2:$};
  \node (6) at (6,-1) [shape=circle, draw] {2};
  \node (7) at (8,-1) [shape=circle, draw] {3} node at (9.5, -1) {$I_3=\{s_3\}$};
  \node (8) at (7,-2) [shape=circle, draw] {4};

  \path[every node/.style={font=\small}]
    (6) edge [bend left] (5)
    (7) edge [bend right] (5)
    (8) edge [bend right] (7);
\end{tikzpicture}
\end{center}

Note that as agent $3$ does not receive signal from agent $4$ at $t=1$, she does not have any new information to send to agent $1$ at the second communication round.  Therefore, agent $1$'s information set is still $\{s_1, s_2, s_3\}$ at $t=2$.  By contrasting the two cases in this toy example, we see that agents' decisions affect the information flow in the network.

%\subsection{Information Exchange Game}

Now we formally introduce the information exchange game.
Suppose we are interested in a social network with agents $\mathcal{N}^n=\{1,2,...,n\}$.
To model communication in the network, we organize these agents in a directed graph $G_n=(\mathcal{N}^n, \mathcal{E}^n)$, in which each node $i\in\mathcal{N}^n$ represents an agent.
We allow directed graphs to have multi-edges, so that two agents can communicate to each other.
%Given $G_n$, the following game regarding information exchange will be refereed to as an information exchange game $\Gamma_{\info}(G_n)$.
An ordered pair $(j,i)\in\mathcal{E}^n$ means agent $j$ can send information to agent $i$ directly.
The goal of every agent is to estimate $\theta\in\mathbb{R}$, which represents an underlying state of the world.  Agents' knowledge of $\theta$ is captured by a normally distributed common prior $\theta\sim N(0, 1/\rho)$.
At time $t=0$, agent $i$ is endowed with her private signal $s_i = \theta + z_i$.  All $z_i\sim N(0, 1/{\bar\rho})$ are independent and they are also independent of $\theta$.
The distributions of $z_i$'s are common knowledge and so is the network architecture. Our results are not affected  if the means of $\theta$ and $z_i$ are changed to non-zero values.

In this network,  agents exchange their information as follows.
Suppose agents live in a world with continuous time $t\in[0,\infty)$.
Waiting induces a common exponential discount of the payoff with rate $r>0$. Instead of communicating at fixed times,
all agents communicate simultaneously at some points in time that follow a homogeneous Poisson process with rate $\lambda>0$, which is independent of $\theta$ and $z_i$.
This Poisson clock is also common knowledge.
After communication, agents update beliefs according to the Bayes' rule.
For example, the posterior distribution of $\theta$ on $k$ distinct signals is Gaussian with precision $\rho+k\bar\rho$.
So more private information, i.e., a higher $k$, will increase the precision and lead to a better estimate.
Hence, there is a natural trade-off between waiting to get more information and acting earlier to reduce the discount of information value, which makes an optimal stopping problem for each agent $i$.
We call the incentive to get more information \textit{information effect}, and the incentive to act earlier \textit{discount effect}.
In this course, at any given time $t$, each agent $i$ either makes an estimate $x_i$ of the fundamental state of the world $\theta$, or ``wait" for more information.
Just as illustrated in the four agents' example,  we assume that after agents make estimate and exit, they do not receive new information, but they continue to transmit information that they have already obtained when new rounds of communication take place.  %Upon transmitting all her information once, the agent can quit the game as she receives no new information to pass along.

We introduce a few more notations to facilitate the discussion.
Let $I_{i,t}^n$ denote the information set of agent $i$ at time $t$.
We next specify the payoff structure and the optimization problem faced by agents.
%Different from \cite{ABO2011}, we employ a recursive definition of agents' problem, which highlights agents' actual optimizing behavior. %\textcolor{red}{``validity of what"?  And what is "principle of optimality"?}
Suppose agent $i$ takes action $x_i$ at time $t$ when the realization of the underlying state is $\theta$, then
her instantaneous payoff of taking an action $x_i$ is
$$u_i^n(x_i)= \psi - (x_i-\theta)^2\,,$$
where $\psi$ is a real-valued constant that captures the information sensitiveness of the decision problem, which we will elaborate later.  %\textcolor{red}{It is worth noting that in \cite{ABO2011}, the constant $\psi$ in agents' instantaneous payoff function is implicitly assumed to satisfy a constraint $\psi(\rho+\bar{\rho})>1$. We relax this assumption and allow $\psi$ to take any real value. - Lao Lin: Probably no need to mention this anymore. }  %We will discuss the importance of doing this below.
At time $t$ with information set $I_{i,t}^n$, agent $i$'s optimal expected instantaneous payoff of taking an action before discounting is
$$U_{i,t}^n(I_{i,t}^n)=\max_{x_i}\mathbb{E}(u_i^n(x_i)|I_{i,t}^n)\,.$$
It is easy to see that agent $i$'s optimal estimate is $x_{i,t}^{n,*}=\mathbb{E}[\theta|I^{n}_{i,t}]$ if she decides to act at time $t$.
Thanks to the normality assumption of the fundamental $\theta$ and signals $\{s_i\}_{i=1}^n$, the optimal expected instantaneous payoff of agent $i$ taking an action after observing $k$ distinct signals can be calculated explicitly:
\begin{equation} \label{jf1}
\mathbb{E}[\psi - (x^{n,*}_{i,t}-\theta)^2|I^n_{i,t}]=\psi - \frac{1}{\rho+\bar\rho k}\,.
\end{equation}

At any time $t$ with information set $I_{i,t}^n$, before trying to make a best estimate and exit, agent $i$ has to make a decision about whether to exit. To facilitate the analysis, we first assume that any agent can obtain non-negative payoff upon her exit.  This assumption will be formally characterized after we define the equilibrium.  As a result, due to discount in time, each agent should make an estimate and exit precisely at a finite time, and especially, at a time instantaneously after communications take place.  Moreover, each agent would only get finite number of signals even if they waited forever, because there are in total $n$ signals
$\{s_i\}_{i=1}^n$ in the network.  Therefore, we actually only need to consider strategy profiles in which every agent exits at a finite communication round, rather than at any arbitrary time.
%The above analysis allows us to consider an essentially equivalent game that is much simpler than what is studied in \cite{ABO2011}, without sacrifice in economic and social intuitions. 
Denote by $l^n=(l^n_1, \ldots, l^n_n)$, where each $l^n_i$ is agent $i$'s communication round before exit.   Throughout the paper, we use $l_{-i}^n$ to denote $l^n$ without the component $l_i^n$.  Let $\tau_{k}$ be the physical time until $k$ rounds of communication.  Agent $i$'s payoff for choosing action $l_i^n$ is
$$
U_i^n(l^n_i, l^n_{-i})=\mathbb{E}\left\{ e^{-r\tau_{l^n_i}}\max_{x_i}\mathbb{E}[\psi - (x_i - \theta)^2|I_i^n(l^n)]\right\}\,,
$$
where $I_i^n(l^n)$ is agent $i$'s information set upon exit, which depends on other agents' exit strategies $l^n_{-i}$.  By (\ref{jf1}) and the exponential waiting time of the Poisson clock, we have
$$
U_i^n(l_i^n, l^n_{-i}) = \bar{r}^{l^n_i}\left(\psi - \frac{1}{\rho + \bar\rho k_i^{n, l^n}}\right),
$$
where $\bar{r} = \lambda/(\lambda+r)$ and
$k_i^{n,l^n}$ is the number of signals agent $i$ get upon exit if every agent acts according to $l^n$ in the network $G_n$.  With this reduction,  the following complete information static game will be considered.

\begin{defi} \label{game}
The information exchange game $\Gamma_{\info}(G_n)$ is a triple $\{\mathcal{N}^n, \mathcal{L}^n, \mathcal{U}^n\}$, in which\\
(a) $\mathcal{N}^n$ is the set of agents, i.e., $\mathcal{N}^n=\{1,2,...,n\}$;\\
(b) $\mathcal{L}^n$ is the collection of agents' strategy spaces.  For any agent $i\in \mathcal{N}^n$, her strategy space $L_i^n\in \mathcal{L}^n$ is a finite set
$$L_i^n=\{0,1,2,...,(L_i^n)_{max}\}\,,$$ where
$
(L_i^n)_{max}=\max_{j\in G_n} \{\text{length of shortest path from j to i}\}\,;
$\\
(c) $U^n_i\in\mathcal{U}^n$ is the payoff function for agent $i$:
\begin{equation} \label{payofffunction}
U_i^n(l^n_i, l^n_{-i})=\bar{r}^{l_i^n}\left(\psi-\frac{1}{\rho+\bar{\rho}k_i^{n, l^n}}\right)\,.
\end{equation}
\end{defi}

We consider pure-strategy Nash equilibria of this game.  As an agent's payoff gain from waiting is weakly larger (i.e., no smaller than) when other agents also wait more rounds, the information exchange game is a supermodular game.  The following result is a direct application of \cite{T1979}, which guarantees the existence of a pure-strategy Nash equilibrium in  supermodular games.

\begin{lem}
The information exchange game $\Gamma_{\info}(G_n)$ has at least one pure-strategy Nash equilibrium.
\end{lem}

We denote a pure-strategy Nash equilibrium of the game by $\sigma^{n,*}$, and the set of all pure-strategy Nash equilibria by $\Sigma^{n,*}$.  We further denote by $l_i^{n,\sigma^*}$ the communication steps after which agent $i$ exits under equilibrium $\sigma^{n,*}$, and denote by $k_i^{n,\sigma^*}$ the number of distinct signals agent $i$ has obtained when she exits under equilibrium $\sigma^{n,*}$.  In order to make sure that every agent $i$ gets non-negative payoffs and exits ultimately in the initial strategic circumstance of information exchange, we focus on information exchange games and associated equilibria that satisfy the following assumption in the rest of this section.

\begin{ass} \label{alwaysexit}
$\psi[\rho+\bar{\rho}(k_i^{n,\sigma^*})_{max}]\geqslant 1$ for all agent $i$, where $(k_i^{n,\sigma^*})_{max}$ is the maximum number of signals agent $i$ can get if all other agents choose their exit steps according to $\sigma^{n,*}$.
\end{ass}

The parameter $\psi$ captures the information sensitiveness of the decision problem.  Interestingly, the information sensitiveness of the decision problem is not monotone in $\psi$.  When $\psi$ takes negative or very small positive value, agents would like to wait forever to discount payoff to zero, in which case the decision problem is information irrelevant.  When $\psi$ is large enough, information is relevant.  Specifically, when $\psi$ is moderate, information effect dominates, and thus the decision problem is more information sensitive; while when $\psi$ is large, the discount effect dominates, and thus the decision problem is less information sensitive.  %We therefore hold that discussion of information aggregation in networks should be broken down to information sensitive or less information sensitive cases.  In the analysis below, we will see that information sensitiveness of the decision problem plays a role in shaping the information exchange process and the level of information aggregation.

Now we provide an example of the network game and its equilibrium.  On the four-agent graph displayed previously, suppose $\lambda = r$, $\psi = 1$ and $\rho = \bar\rho = \frac{1}{2}$.  The decision problem for agents $2$ and $4$ are simple.  They should exit right away because they will not get any new signals due to graph structure, but incur discounting penalty should they not act promptly.  The payoff matrix for agent $1$ (row) and $3$ (column) is as follows, in which the first and the second value in each cell are respectively the payoffs of agent 1 and agent 3 [see \eqref{payofffunction}].

\begin{center}
\def\temptablewidth{0.38\textwidth}
\begin{tabular*}{\temptablewidth}{cc|c|c}
\multicolumn{4}{c}{\qquad\qquad\qquad Agent 3}\\
& & 0 Step & 1 Step\\
 \hline
& 0 Step & 0, \text{ }0 & 0,\text{ }$\frac{1}{6}$\\\cline{2-4}
Agent 1 & 1 Step & $\frac{1}{4}$,\text{ }0 & $\frac{1}{4}$,\text{ }$\frac{1}{6}$\\
\cline{2-4}
& 2 Step & $\frac{1}{8}$,\text{ }0 & $\frac{3}{20}$,\text{ }$\frac{1}{6}$
\end{tabular*}
\end{center}
There is one equilibrium of the game.  In this equilibrium, agents $2$ and $4$ exits immediately after they receive their private signals, while agent $1$  and agent $3$ exit after the first communication round.

Before proceeding to discuss information aggregation or learning status, we briefly discuss the equilibrium outcomes of the strategy game in Definition \ref{game}.  This reduced game is a complete information static game, which involves no uncertainty.  However, the uncertainties in the fundamental and in the communication clock were abstracted out through taking expectations, which results in the deterministic payoff function (\ref{payofffunction}).  Therefore the two equilibrium outcomes, $l_i^{n,\sigma^*}$ and $k_i^{n,\sigma^*}$, both deterministic, characterize the strategic interactions of information exchange among agents in the initial circumstance.  This enables us to characterize a learning status by focusing only on such equilibrium outcomes.

We can perform the following comparative statics of the number of signals agent $i$ obtains under equilibrium $k_i^{n,\sigma^*}$.  Intuitively, $k_i^{n,\sigma^*}$ is larger when the discount rate is smaller or the Poisson clock is faster.  It is also larger when the precision of public information $\rho$ is lower or the decision problem is more information sensitive.  However, the precision of private information $\bar\rho$ has ambiguous impact on $k_i^{n,\sigma^*}$, because an increase in the precision of private information has two conflicting effects.  It increases not only the relative quality of the private signal at hand, which prompts an agent to exit earlier, but also the relative information content of her neighbors' private signals, which in turn encourages her to wait.  The former effect is stronger when the precision of public information is higher, while the latter is stronger when the precision of public information is lower.  The discussions in this paragraph can be formalized against mathematical rigor, but the game-theoretic technicality involved is beyond the scope of this paper.

%\begin{lem} \label{k:cs}
%Suppose Assumption \ref{alwaysexit} holds.  Given the communication network $G_n$, the equilibrium outcome $k_i^{n,\sigma^*}$ is weakly increasing in $\lambda$ and weakly decreasing in $r$, $\rho$, and $\psi$ for any agent $i$.  The impact of $\bar\rho$ on $k_i^{n,\sigma^*}$ is ambiguous and depends on other parameters.  In particular, $k_i^{n,\sigma^*}$ is weakly decreasing in $\bar\rho$ if $\rho\leqslant \bar\rho k_i^{n,\sigma^*}$.
%\end{lem}

Finally, we also remark that the role of $k_i^{n,\sigma^*}$ is our paper is similar to the influence vector $v$ in \cite{ACOT2011}.  The quantity $k_i^{n,\sigma^*}$ will play a central role in the next sections.

\section{Finite Population Learning}
\setcounter{equation}{0}
In this section, we measure the level of information aggregation in any given communication network.  Related recent research on learning in social networks focuses on asymptotic learning, which means that as the fraction of agents taking the correct action converging to one as the population of the social network grows large (\cite{ADLO2011, ABO2011}).
However, as discussed in \cite{AO2010}, people are also interested in the information dynamics away from long run limit. 
In pursuing this goal, a new concept of learning in social networks is introduced.  
\begin{defi} \label{fpl}
Given a social network $G_n$, the information exchange game $\Gamma_{\info}(G_n)$ and an equilibrium profile $\sigma^{n,*}$, for a triple $(\varepsilon, \bar\varepsilon, \delta)$, we say $G_n$ achieves $(\varepsilon, \bar\varepsilon, \delta)$-learning under $\sigma^{n,*}$ if
$$
\mathbb{P}_{\sigma^{n,*}}\left(\frac{1}{n}\sum_{i=1}^n \left(1-M_i^{n,\varepsilon}\right)\geqslant \bar\varepsilon\right)\leqslant \delta\,,
$$
where $M_i^{n,\varepsilon} = \textbf{1} (|x_i-\theta|\leqslant\varepsilon)$, $x_i$ is agent $i$'s optimal action upon exit, and $\mathbb{P}_{\sigma^{n,*}}$ denotes the conditional probability given $\sigma^{n,*}$\,.
\end{defi}

In this definition, the parameter $\varepsilon$ sets the precision on what the approximately correct decision is for individual agents, $1-\bar\varepsilon$ controls the fraction of agents who make the approximately correct decision, and $1-\delta$ represents the probability at which such a high fraction of agents make the approximately correct decision.  In particular, we highlight the difference between $\varepsilon$ and $\bar{\varepsilon}$, because these two parameters capture different tolerances. Concretely, $\varepsilon$ is at the individual level while $\bar{\varepsilon}$ is at the aggregate level.  %These two tolerances are not necessarily equal, and we will highlight the importance of distinguishing them below.

%An important difference between our definition of $(\varepsilon, \bar\varepsilon, \delta)$-learning and the ($\varepsilon$, $\delta$)-asymptotic learning in \cite{ABO2011} is that, while their learning concept is defined on a society $\{G_n\}_{n=1}^\infty$ that is an infinite sequence of communication networks, ours is defined on one social network $G_n$ with finite population $n$.  This single network based definition allows us to study whether learning occurs in a given social network, besides  exploring the long run limit behavior along the society.  To capture this essence, we call the new learning concept \emph{finite population learning} in verbal discussions.

A natural question to ask is whether such finite population learning occurs in a given communication network. If so, under what conditions?  The following proposition provides a necessary condition and a sufficient condition for $(\varepsilon, \bar\varepsilon, \delta)$-learning in a given social network under any equilibrium profile.  When there is no confusion, we refer to the information exchange game $\Gamma_{\info}(G_n)$ simply as $G_n$. Denote by $\erf(x)=\frac{2}{\sqrt{\pi}}\int_{0}^x e^{-t^2} dt$ the error function of the standard normal distribution.
\begin{prop}\label{fini}
For a given social network $G_n$ under any equilibrium  $\sigma^*(=\sigma^{n,*})$,\\
(a) $(\varepsilon, \bar\varepsilon, \delta)$-learning does not occur if
\begin{equation} \label{fininece}
\frac{1}{n}\sum_{i=1}^n \erf\left(\varepsilon\sqrt{\frac{\rho +  \bar\rho k_i^{n,\sigma^{*}}}{2}}\right)<(1-\bar\varepsilon)(1-\delta)\,.
\end{equation}
(b) $(\varepsilon, \bar\varepsilon, \delta)$-learning occurs if
\begin{equation} \label{finisuff}
\frac{1}{n}\sum_{i=1}^n \erf\left(\varepsilon\sqrt{ \frac{\rho + \bar\rho k_i^{n,\sigma^{*}}}{2}}\right)\geqslant 1-\bar\varepsilon\delta\,.
\end{equation}
\end{prop}

%\textcolor{red}{Ming's suggestion: Give others, what precision $\varepsilon$ can I get?  Not necessary to bisection learning or not.}

This proposition provides clear conditions for the occurrence of finite population learning.  Our conditions are more operative and transparent than their asymptotic counterparts in previous literature.  
%Recall that in the Proposition 1 of \cite{ABO2011}, two equilibrium-specific variables, the $k$-radius set $V_k^{n,\sigma}$ and an extra variable $k$, are needed to characterize the conditions, while in our conditions we characterize the finite population learning with only one equilibrium outcome $k_i^{n,\sigma^*}$.  To check conditions in \cite{ABO2011}, one need to first find a $k$ that satisfies the condition for the error function, and then check if the set $V_k^{n,\sigma}$ satisfies the limit condition under equilibrium.  In this process, the final determination of $k$ and its relation to the equilibrium are unclear.
Specifically, our conditions only require one equilibrium outcome $k_i^{n,\sigma^*}$, and the set $\{k_i^{n,\sigma^*}\}_{i=1}^n$ is directly induced by an equilibrium $\sigma^{n,*}$ in a communication network $G_n$.  Hence, conditions \eqref{fininece} and \eqref{finisuff} not only allow us to investigate the effect of learning in a given communication network, but also offer a more interpretable link between the communication equilibrium and its corresponding information aggregation status.

Conditions \eqref{fininece} and \eqref{finisuff} also allow us to untangle the interplay among parameters.
For example, we are able to answer the following question.
Given the tolerances $\varepsilon$, $\bar{\varepsilon}$, $\delta$ and the information precisions $\rho$ and $\bar{\rho}$, how does the change of $k_i^{n,\sigma^*}$ affect the occurrence of finite population learning in a given social network $G_n$?
When $k_i^{n,\sigma^*}$'s are sufficiently small to validate condition (\ref{fininece}), finite population learning does not occur.
Similarly, when most of $k_i^{n,\sigma^*}$'s are sufficiently large so that the condition (\ref{finisuff}) is satisfied, finite population learning occurs.  Similar marginal interpretations also apply to parameters $\varepsilon$, $\bar{\varepsilon}$, $\delta$, $\rho$ and $\bar{\rho}$.   Generally, finite population learning in a given social network $G_n$ is more likely to occur when the equilibrium induces larger numbers of signals obtained by agents. It is also more likely to occur when the tolerances and the information precisions are higher.  As  interplays among the parameters $\varepsilon$, $\bar{\varepsilon}$, $\delta$, $\rho$,  $\bar{\rho}$ and $k_i^{n,\sigma^*}$ are clear through \eqref{fininece} and \eqref{finisuff}, the two conditions provide various comparative statics that help us better understand learning in different social circumstances.
%Finally, recall Corollary \ref{k:cs}, which shows $k_i^{n,\sigma^*}$ is weakly increasing in $\lambda$, weakly decreasing in $r$, $\rho$ and $\psi$, and weakly decreasing in $\bar\rho$ if $\rho\leqslant \bar\rho k_i^{n,\sigma^*}$, the comparative statics could be further extended to include parameters of the information exchange game.
Since the total amount of information is fixed in any finite population network, these comparative statics indeed disentangle the effectiveness of information aggregation from the endowment of information, so that the net effect of information aggregation is transparent.

%It is interesting to highlight the roles of precisions of public information $\rho$ and private information $\bar\rho$ in shaping the status of finite population learning, which also serves as a good example to demonstrate the power of these conditions in Proposition \ref{fini}.

%We would like to emphasize that Lemma \ref{k} provides a solid foundation for Proposition \ref{fini}.  %The importance of the  uniquely deterministic equilibrium outcome $k_i^{n,\sigma^*}$ as discussed in previous section is also displayed here.
%Suppose $k_i^{n,\sigma^*}$'s were not ex-ante uniquely deterministic, the left hand sides of the inequalities would be random when we check the conditions in Proposition \ref{fini}. It is possible that different realization of $k_i^{n,\sigma^*}$ leads to different evaluations of these conditions.
%Hence, were $k_i^{n,\sigma^*}$'s ex-ante random, Proposition \ref{fini} is not a good criterion to check finite population learning.
%In \cite{ABO2011}, the authors have not explicitly proved that the two equilibrium-specific variables  $V_k^{n,\sigma}$ and $k$ are ex-ante uniquely deterministic although they are actually so.\footnote{Following the idea of our Lemma \ref{k}, we can formally prove that $V_k^{n,\sigma}$ and $k$ in the Proposition 1 of \cite{ABO2011} are ex-ante uniquely deterministic.}

It is interesting to note that $(1-\bar{\varepsilon})(1-\delta)<1-\bar{\varepsilon}\delta$ for any $0<\bar{\varepsilon}, \delta<1$.  This gap indicates that failure of condition (\ref{fininece}) does not necessarily lead to condition (\ref{finisuff}), and vice versa.  %Therefore, there exist circumstances under which the occurrence of finite population learning is undetermined in view of conditions \eqref{fininece} and \eqref{finisuff}. %\footnote{Such gap is also involved in the Proposition 1 of \cite{ABO2011}.  But as discussed before, as their necessary and sufficient conditions for asymptotic learning is not transparent enough, this gap cannot be explicitly displayed there.}
Two perspectives help understand this gap.  First, we use Markov's inequality to get tractable forms of the necessary and the sufficient conditions.  Sharper inequalities may lead to weaker conditions and thus probably fill a part of the gap, but they are likely to make these conditions intractable and less transparent.  Secondly and more importantly, as we discussed above,  conditions \eqref{fininece} and \eqref{finisuff} involve equilibrium outcomes in a clean and simple formula.
%Therefore, to determine finite population learning occurrence, we do not need full information of the equilibrium profile.
The cost for enjoying this clarity is that we did not fully utilize $\{k_i^{n,\sigma^*}\}_{i=1}^n$.

%\textcolor{red}{still need to think more on the second point.  is the equilibrium variable set 1-1 with the equilibrium? A family of equivalent equilibria; complete graph as an example.}

Also, a beauty of symmetry arises in our necessary and sufficient conditions for finite population learning. The parameters $\bar\varepsilon$ and $\delta$ are completely interchangeable in these conditions, which was not expected as they captures tolerances in different categories.  On the other hand, in our two conditions, parameter $\varepsilon$ stands in a position that is unchangeable with $\bar\varepsilon$ and $\delta$, which hints that $\varepsilon$ and $\bar\varepsilon$ play different roles in finite population learning. %As a result, it is more confirmed that to distinguish between $\varepsilon$ and $\bar\varepsilon$ is theoretically non-trivial in the definition of finite population learning.

%\textcolor{red}{phase diagram here.}

Conditions \eqref{fininece} and \eqref{finisuff} have powerful implications.  The next corollary establishes a necessary condition and a sufficient condition without equilibirum outcomes.  The proof is straightforward, but the results are non-trivial.  %The key is again that our conditions \eqref{fininece} and \eqref{finisuff}  only involve  $k_i^{n, \sigma^*}$ and we know that $1\leqslant k_i^{n, \sigma^*}\leqslant n$.

\begin{coro} \label{finicoro}
For any social network $G_n$ and any equilibrium $\sigma^{n,*}$,\\
(a) $(\varepsilon, \bar\varepsilon, \delta)$-learning does not occur if
\begin{equation} \label{finicoronece}
\erf\left(\varepsilon\sqrt{\frac{\rho+\bar\rho n}{2}}\right)<(1-\bar\varepsilon)(1-\delta)\,.
\end{equation}
(b) $(\varepsilon, \bar\varepsilon, \delta)$-learning occurs if
\begin{equation} \label{finicorosuff}
\erf\left(\varepsilon \sqrt{\frac{\rho+\bar\rho}{2}}\right)\geqslant 1-\bar\varepsilon\delta\,.
\end{equation}
\end{coro}

%For any social network $G_n$ and equilibrium $\sigma^{n,*}$, each agent can get at most $n$ signals before she takes an action; in other words, $k_i^{n,\sigma^*}\leqslant n$ for every agent $i$.  
%Hence, the left hand side of condition (\ref{finicoronece}) is an upper bound of the left hand side of condition (\ref{fininece}).  
%This implies that when  condition (\ref{finicoronece}) is satisfied,  condition (\ref{fininece}) is also satisfied for any $G_n$ and $\sigma^{n,*}$, and thus $(\varepsilon, \bar\varepsilon, \delta)$-learning does not occur whatever $G_n$ and $\sigma^{n,*}$ are.  Similarly, because every agent at least has her own private signal, $k_i^{n,\sigma^*}\geqslant 1$ for all $i$, when condition (\ref{finicorosuff}) is satisfied, condition (\ref{finisuff}) must be satisfied too for any $G_n$ and $\sigma^{n,*}$.  Thus,  finite population learning occurs whatever $G_n$ and $\sigma^{n,*}$ are.

Corollary \ref{finicoro} follows from the fact that $1\leq k_i^{n,\sigma^*}\leq n$. It is interesting because under some circumstances, we can determine the occurrence of finite population learning without knowing either the structure of the social network or the equilibrium. Intuitively,  if any one parameter of the tolerances, information precisions or population size is too low, such that the condition (\ref{finicoronece}) is satisfied, we may conclude that finite population learning does not occur no matter how effective the communication network is organized.  Conversely, if any one of the tolerances or information precisions is sufficiently large such that condition (\ref{finicorosuff}) holds, we know that finite population learning surely occurs even if all agents are isolated.  %In this sense, Corollary \ref{finicoro} extends the perspective in the Proposition 2 in \cite{ABO2011} that we can determine learning without knowledge of particular equilibrium, to finite population learning.

Finally, as the information exchange game exhibits strategic complementarity, it is expected that multiple equilibria might emerge under some circumstances.  An interesting perspective in investigating finite population learning is to measure the effect of learning against multiple equilibria.  
%As opposed to Proposition \ref{fini} and Corollary \ref{finicoro} and that offer conditions for learning under a given equilibrium or any equilibrium, learning status can also be evaluated against all equilibria.  
We provide the following generalized (conservative) version of finite population learning to accommodate multiple equilibria without equilibrium selection.

\begin{defi}
Denote by $\Sigma^{n,*}=\{\sigma^{n,*}\}$ the set of equilibria of $\Gamma_{\info}(G_n)$.  The $(\varepsilon, \bar\varepsilon, \delta)$-learning occurs if
$$
\sup_{\sigma^{n,*}\in \Sigma^{n,*}}\mathbb{P}_{\sigma^{n,*}}\left(\frac{1}{n}\sum_{i=1}^n(1-M_i^{n,\varepsilon})\geqslant\bar\varepsilon\right)\leqslant \delta\,.
$$
\end{defi}

This definition offers a conservative standard to evaluate finite population learning in the sense that the least favorable equilibrium determines the learning status.
When $\Sigma^{n,*}$ is a singleton, the above definition reduces to Definition \ref{fpl}. The proof of Proposition \ref{fini} can be recycled to derive the next corollary.
\begin{coro}\label{cor:fini_multi}
Given an information exchange game $\Gamma_{\info}(G_n)$,
\begin{itemize}
\item[(a)] $(\varepsilon, \bar\varepsilon, \delta)$-learning does not occur if

\begin{equation*} %\label{eq:multi neg}
\min_{\sigma^{n,*}\in \Sigma^{n,*}}\frac{1}{n}\sum_{i=1}^n \erf\left(\varepsilon\sqrt{\frac{\rho + \bar\rho  k_i^{n,\sigma^*}}{2}}\right)<(1-\bar\varepsilon)(1-\delta)\,.
\end{equation*}

\item[(b)] $(\varepsilon, \bar\varepsilon, \delta)$-learning occurs if
\begin{equation*} %\label{eq:multi suff}
\min_{\sigma^{n,*}\in \Sigma^{n,*}}\frac{1}{n}\sum_{i=1}^n \erf\left(\varepsilon\sqrt{\frac{\rho + \bar\rho  k_i^{n,\sigma^*}}{2}}\right)\geqslant 1-\bar\varepsilon\delta\,.
\end{equation*}
\end{itemize}
\end{coro}

 %\textcolor{red}{I'll add some reference here.}
%\textcolor{red}{In view of this corollary, it seems that the more interesting networks lie in the intersection of negations of the two conditions in this corollary}

\section{Perfect Learning and the Rates}
\setcounter{equation}{0}
Based on the analysis of finite population learning, we consider the effect of information aggregation and learning as population in communication networks grows.  Our approach to address the limiting behavior of learning  is different from asymptotic learning in existing literature (\cite{ADLO2011, ABO2011}).   In particular, we highlight finite population learning as the foundation of asymptotic learning.  Consequently, we are able to check learning status all along the path to the limit, and the probabilistic tolerance parameters naturally induce learning rates.  This concept of learning rate is different from what is employed in \cite{GJ2012b, GJ2012a, GJ2011} that focuses on the time dimension.  %As suggested by \cite{AO2010}, even if asymptotic learning occurs in the long run, people are also interested in investigating the rates at which the long run distribution is reached in different societies with different patterns of population growth.  Our results on learning rates are of the same spirit.

\subsection{Perfect Learning}

%\textcolor{red}{Jackson, 2009, pp.42, discussion on learning, looks non-bayesian models are more tractable.  Address this claim.}

%After investigating the information dynamics and associated learning in a communication network $G_n$ of population size $n$, we now discuss the information aggregation behavior when population grows.
%Similar to asymptotic learning in \cite{ABO2011}, we consider a sequence of communication networks $\{G_n\}_{n=1}^\infty$, which is called a society.  However, different from asymptotic learning in existing literature, we are interested in whether social networks in a society achieve finite population learning under diminishing tolerances as population grows.  

Recall that we have three tolerance parameters $\varepsilon$, $\bar\varepsilon$, and $\delta$ that define $(\varepsilon, \bar\varepsilon, \delta)$-learning. To inquire the limiting behavior in a society $\{G_n\}_{n=1}^\infty$, where existing links are kept when networks grow,  we can focus on one parameter at a time, keeping the other two fixed.   The following definition introduces $\delta$-perfect learning on a given society $\{G_n\}_{n=1}^\infty$.

\begin{defi}
We say $\delta$-perfect learning occurs in society $\{G_n\}_{n=1}^\infty$ under equilibria $\{\sigma^{n,*}\}_{n=1}^{\infty}$ if there exists a vanishing positive sequence $\{\delta_n\}_{n=1}^{\infty}$ such that
$(\varepsilon, \bar\varepsilon, \delta_n)$-learning occurs in $G_n$ under its associated $\sigma^{n,*}$ for all $n$\,.
\end{defi}

%trade-off between different rate sequences?  Fixed the society, the sequence of delta, what's the sequence of varepsilon?  Or fix one rate, the society, trade-off between other two rates?  4 in total, fix two, consider the other two.

%1.  $(\varepsilon, \bar\varepsilon, \{\delta_n\}_{n=1}^{\infty})$  2.$(\{\varepsilon_n\}_{n=1}^{\infty}, \bar\varepsilon, \delta)$   3. $(\varepsilon, \{\bar\varepsilon_n\}_{n=1}^{\infty}, \delta)$. In each of these scenarios, it is better to have $\delta_n$, $\varepsilon_n$ and $\bar\varepsilon_n$ go to $0$ respectively.

%\textcolor{red}{a technical relaxation is to let $n$ start from some integer bigger than 1.  Then the following remark one is not necessary.}

Compared to the perfect asymptotic learning concept in \cite{ABO2011}, our definition of perfect learning is both stronger and more general for the following reasons.  First, we require the networks in the society to achieve learning not only in the limit but also all along the path towards the limit. Second, by focusing on different parameters $\varepsilon$, $\bar\varepsilon$ and $\delta$, we could address three different kinds of asymptotic learning. As discussed in the previous section, these three parameters exhibit different impacts on finite population learning, so that they can play different roles in perfect learning.  %But for the rest of this section, we will focus on $\delta$-perfect learning.
Third and most importantly, this definition allows us to investigate learning rates, which is the focus of the next subsection.

%\textcolor{red}{Others:  $(\varepsilon, \bar\varepsilon_n, \delta_n)$ relate to ``Uncertainty principle" }

%The analysis of perfect learning relies heavily on the properties of $(\varepsilon, \bar\varepsilon, \delta)$-learning uncovered in the previous section.  
%Previously, we fixed $\delta$ and derived sufficient conditions for $(\varepsilon, \bar\varepsilon, \delta)$-learning in a single network.  
%Now we let $\delta$ vary, and would like to find conditions under which there exits a vanishing positive sequence $\{\delta_n\}_{n=1}^\infty$, such that $(\varepsilon, \bar\varepsilon, \delta_n)$-learning occurs in $G_n$ under its associated equilibrium $\sigma^{n,*}$ for all $n$.

In the following, we will derive two sufficient conditions for $\delta$-perfect learning.  The first condition, stated as Proposition \ref{equilibriuminformedlearn}, relies on the equilibrium outcome $k_i^{n,\sigma^*}$.  The second condition, stated as Proposition \ref{socialinformedlearn},  relies only on the formation of the society.  To deliver the first sufficient condition, we define an \textit{equilibrium informed agent} in a society.

%A drop on the ``rate mentality" leads to the following more refined characterization of $k_i^{n,\sigma}$'s.

\begin{defi}[\textsc{Equilibrium Informed Agent}] \label{equilibriuminformed}
For agent $i$ in a given society $\{G_n\}_{n=1}^\infty$, she is equilibrium informed with respect to $\{G_n\}_{n=1}^\infty$ under equilibria $\{\sigma^{n,*}\}_{n=1}^\infty$ if
\begin{equation*}
\lim_{n\to\infty}k_i^{n,\sigma^*}=\infty\,.
\end{equation*}
\end{defi}

An agent has equilibrium informed status means that she enjoys increasing information advantage as population grows. The next proposition offers a sufficient condition for $\delta$-perfect learning.  In a similar spirit, we have a more general sufficient condition, Lemma \ref{LEM: perfect1}, in the supplementary materials. The proof of Proposition \ref{equilibriuminformedlearn}  is omitted as it is a corollary to Lemma \ref{LEM: perfect1}.  

\begin{prop} \label{equilibriuminformedlearn}
The $\delta$-perfect learning occurs in a society $\{G_n\}_{n=1}^\infty$ under equilibria $\{\sigma^{n,*}\}_{n=1}^\infty$ if
%\textcolor{red}{Remark: At some point, we might say that $\sigma$ denotes both the equilibrium of each network $G_n$ and that of the collective equilibria of the society.}
$$
\lim_{n\to\infty}\frac{1}{n}|\EI^{n,*}|=1\,,
$$
where $\EI^{n,*}$ the set of equilibrium informed agents in the network $G_n$ under equilibrium $\sigma^{n,*}$.
\end{prop}

Proposition \ref{equilibriuminformedlearn} states that perfect learning occurs when almost all agents are equilibrium informed.  This is consistent with the idea of social learning that successful learning allows individuals to have sufficient information to make a good decision,  and that such individuals represent an overwhelming proportion of the society.
We can have such a transparent condition because our perfect learning concept is powered by finite population learning, a sufficient condition of which only involves one set of equilibrium variables:  $\{k_i^{n,\sigma^*}\}$.

%is still the mentality of finite population learning.  Finite population learning ensures that only one equilibrium outcome $k_i^{n,\sigma^*}$ is sufficient for checking the result of learning in any given communication network.  Furthermore, our concept of perfect learning is based on finite population learning along the interested society, so that the clarity of characterizing perfect learning naturally is inherited from the nature of finite population learning.

Next we consider the second sufficient condition that relies only on formation of the society. To streamline the presentation in the main texts, we assume that each agent enjoys a positive payoff even if she exits at the beginning.  
%\footnote{In \cite{ABO2011}, this assumption is implicitly imposed throughout the analysis.}
From (\ref{jf1}), this is equivalent to the following assumption.

\begin{ass}\label{Assumption:regular}
$(\rho+\bar{\rho})\psi>1$\,.
\end{ass}

We will hold this assumption for the rest of this section. In the supplementary materials, we relax this assumption and discuss all possible cases.

Before looking into the next sufficient condition for perfect learning, we first point out an important observation which states that, although the number of signals an agent gets in equilibrium may diverge to infinity with growth of the communication network, the equilibrium communication steps will not increase unboundedly.

\begin{lem} \label{round short}
Under Assumption \ref{Assumption:regular}, for any agent $i$, the communication rounds she optimally experiences before taking an action in any social network $G_n$ is bounded from above by a constant independent of $n$. Mathematically,
\begin{equation} \label{roundshortbound}
l_i^{n, \sigma^*}\leqslant l_i^n<\min\left\{(L_i^n)_{max}, \ln\left(1-\frac{1}{(\rho+\bar{\rho})\psi}\right)/\ln\bar{r}\right\}\,,
\end{equation}
in which $l_i^n$ stands for the optimal communication rounds for agent $i$ given that other agents wait until the maximum allowable step, and $(L_i^n)_{max}$ is the maximum length of all paths ended with $i$ in $G_n$.
\end{lem}

A more general version (without Assumption \ref{Assumption:regular}) of Lemma \ref{round short} with its associated proof is included in the supplementary materials as Lemma \ref{round}.  A key idea behind Lemma \ref{round short} is that after agent $i$ gets sufficiently large number of signals within some finite communication steps, even expecting infinite number of signals does not justify the discount of further waiting.  The intuition is that for well connected agents, they will get sufficient information after a few communication rounds to make a decision, whereas for the not well connected agents, waiting too long discounts their information value. From condition (\ref{roundshortbound}), we see that the upper bound is exclusively determined by parameters of the information exchange game.
%Specifically, the upper bound is increasing in $r$ while decreasing in $\rho$, $\bar\rho$, $\psi$ and $\lambda$.

Lemma \ref{round short} plays an important role in shaping our next sufficient condition that bypasses equilibrium and directly links perfect learning to network formations.  Recall Proposition \ref{equilibriuminformedlearn} which states that almost all agents' $k_i^{n,\sigma^*} \to \infty$ is sufficient for perfect learning.  On the other hand, from  Lemma \ref{round short} we know that no agent has an optimal unbounded communication step $l_i^{n,\sigma^*}$.  By combining the two observations, we know the only possibility to validate Proposition \ref{equilibriuminformedlearn} is that almost all agents get unbounded number of signals within finite communication steps.
%In other words, if an agent obtains more signals along the society, it should not be through waiting for increasingly more communication steps.
This consideration leads to our following definition of a \textit{socially informed agent}.

\begin{defi}[\textit{Socially Informed Agent}] \label{socialinformed}
For each agent $i$ in a given society $\{G_n\}_{n=1}^\infty$, let $L_i = \min\{l_0\in\mathbb{N}: \lim_{n\rightarrow\infty}|B^n_{i,l_0}|=\infty\}$, where $B^n_{i,l}$ is the set of agents in $G_n$ whose shortest path to $i$ has length at most $l$. Agent $i$ is socially informed with respect to $\{G_n\}_{n=1}^\infty$ if $L_i$ is finite, and if there exists $N\in\mathbb{N}$ such that for $n\geqslant N$, we have
\begin{equation} \label{positiveinform}
\psi-\frac{1}{\rho+\bar{\rho}|B_{i,L_i}^n|}>0\,,
\end{equation}
and
\begin{equation} \label{incentiveinform}
\bar{r}^{L_i}\left(\psi-\frac{1}{\rho+\bar{\rho}|B_{i,L_i}^n|}\right) > \bar{r}^l\left(\psi-\frac{1}{\rho+\bar{\rho}|B_{i,l}^n|}\right)\\ \text{ for all } 0 \leq l< L_i\,.
\end{equation}
Moreover, we denote by $\SI^n$ the set of socially informed agents in the network $G_n$.
\end{defi}

In Definition \ref{socialinformed}, condition (\ref{positiveinform}) is automatically satisfied in view of Assumption \ref{Assumption:regular}.  Intuitively, a socially informed agent can be reached by a large number of neighbors after some finite communication steps $L_i$.  Furthermore, condition (\ref{incentiveinform}) ensures that this agent strictly prefers to wait at least until the arrival of such communication step $L_i$, given other agents never exit.  Therefore, agent $i$ is guaranteed to obtain a large number of signals from finite communication steps, if other agents never exit.  Note also that the definition of a socially informed agent does not require knowledge of any specific equilibrium.  It only depends on the topological structure of the graph and on the parameters in the information exchange game. With the help of socially informed agents, we bypass equilibrium and state the following sufficient condition for perfect learning.

%Agent $i$ is socially informed with respect to a society $\{G_n\}_{n=1}^\infty$, if there exists $S_i$, such that $S_i$ is the smallest positive integer to satisfy
%\begin{equation*}
%\lim_{n\to\infty} |B^n_{i,S_i}| = \infty\,,
%\end{equation*}
%and there exists $N$ such that for all social networks $G_n\in \{G_n\}_{n=1}^\infty$ with $n\geqslant N$, there are
%\begin{equation*}
%\bar{r}^{S_i}\left(\psi-\frac{1}{\rho+\bar{\rho}|B_{i,S_i}^n|}\right)>0
%\end{equation*}
%and
%\begin{equation} \label{incentiveinform}
%\bar{r}^{S_i}\left(\psi-\frac{1}{\rho+\bar{\rho}|B_{i,S_i}^n|}\right)\geqslant \bar{r}^s\left(\psi-\frac{1}{\rho+\bar{\rho}|B_{i,s}^n|}\right) \text{ for all }s\leqslant S_i, s\in\mathbb{N}\cup\{0\}\,.
%\end{equation}
%\textcolor{red}{Moreover, we denote the infinite sets of agents who are socially informed in $\{G_n\}_{n=1}^\infty$ as $\{\SI^n\}_{n=1}^\infty$.}

\begin{prop} \label{socialinformedlearn}
The $\delta$-perfect learning occurs in a society $\{G_n\}_{n=1}^\infty$ under any equilibrium $\{\sigma^{n,*}\}_{n=1}^\infty$ if
%\textcolor{red}{Remark: At some point, we might say that $\sigma$ denotes both the equilibrium of each network $G_n$ and that of the collective equilibria of the society.}
$$
\lim_{n\to\infty}\frac{1}{n}|\SI^n|=1\,.
$$
\end{prop}

Proposition \ref{socialinformedlearn} is interesting because we can determine the occurrence of perfect learning through knowledge on the formation of society alone.  Given the tractable conditions for socially informed agents, we can check whether a given society sufficiently supports perfect learning under any equilibrium.  Especially, given the difficulty of explicitly solving for equilibria of the information exchange game in general cases, Proposition \ref{socialinformedlearn} is of more value. 
%The learning results of \cite{ABO2011}, based on asymptotic learning, has similar advantage as our Proposition \ref{socialinformedlearn}, but our sufficient condition is more transparent and powerful.  We can do so largely because the key observation documented in Lemma \ref{round short} is uncovered that no agent would like to wait too long under any equilibrium .
%Specifically, their condition (5) involves double limits of $|V_k^n|$, which is not easy to interpret in a natural way.
%The reason that \cite{ABO2011} did not get conditions like ours is probably because they have not explicitly pointed out results similar to our Lemma \ref{round short}, so that they have to consider some cases that cannot occur.
%\textcolor{red}{This paragraph should be redone. In particular, make less contrast between this paper and Acemoglu's, but put some interpretation instead.}

%Finally, by the condition \eqref{incentiveinform}, notice $k_i^{n,\sigma}\geq |B_{i,S_i}^{n,\sigma}|$ for sufficiently large $n$.

%%To see this by a contradiction, suppose for any possible path $\{j_{S_i-1}, j_{S_i-2}, ..., j_1, i\}$ from any $j_{S_i-1}\in B_{i, S_i-1}^n$ to $i$, there exists $s$ such that.

\subsection{Learning Rates}
In this subsection we define the learning rate for $\delta$-perfect learning.  %Without loss of generality, we still focus on $\delta$-perfect learning.
It is natural to expect similar concepts of learning rates for $\varepsilon$-perfect learning and $\bar\varepsilon$-perfect learning.

\begin{defi}
If $\delta$-perfect learning occurs in $\{G_n\}_{n=1}^\infty$ under equilibria $\{\sigma^{n,*}\}_{n=1}^{\infty}$, then we call the corresponding sequence of tolerances $\{\delta_n\}_{n=1}^{\infty}$ the learning rate.
\end{defi}

%Intuitively, a desirable concept of learning rates must capture the quality of perfect learning.  More concretely, faster learning rate should imply higher quality of perfect learning.  Therefore, the sequence of tolerances themselves has natural advantages over other candidates.  %Again recall that perfect learning is defined based on finite sample learning, these advantages should also be viewed in our both concepts of learning.

%Recall that in finite population learning, the three tolerance parameters are natural measures of learning quality.
%With $\varepsilon$ and $\bar\varepsilon$ fixed, $\delta$ should be considered as the unique measure of learning quality, and smaller $\delta$ indicates better quality of finite population learning.  Regarding perfect learning, learning rate as a measure of learning quality should be defined globally along the path to the limit, so any point measure is undesirable.  Hence, a faster convergence rate of the vanishing sequence  $\{\delta_n\}_{n=1}^{\infty}$ indicates higher quality of perfect learning.% which is represented as faster learning rate.

It is worth highlighting the difference between our learning rate concept and the speed of convergence to a pre-defined consensus mainly employed in observational learning problems. The latter concerns about the time towards a consensus in a circumstance where people make repeated decisions and learn from others' previous decisions to help make their own future decisions (\cite{GJ2012b, GJ2012a, GJ2011}).  In observational learning problems with repeated decisions, the observable sequence of aggregate decisions naturally reveals the time dynamics of information aggregation.  However in direct communication setup,  dynamics of information aggregation along the time dimension is largely unobservable, which calls for alternative dimensions to look into the information dynamics. Tolerance parameters $\{\delta_n\}_{i=1}^{\infty}$ offer a natural standpoint to look into the information aggregation dynamics.  This feature also distinguishes our work from existing literature on learning rate of similar spirit.  For example,  \cite{ADLO2009} attempt to define and investigate an asymptotic learning based rate in an observational learning context.  Although their concept also captures a sequence of diminishing probabilities, it does not characterize the learning status in every social network along the society.

%An answer to this question necessarily arises from very crude approximation of $k_i^{n,\sigma}$, in view of our Proposition 4.

%The previous subsection concerns about existence of $\delta$-perfect learning.  In this subsection, we  focus on pinning down specific learning rates $\{\delta_n\}_{n=1}^{\infty}$.  More generally, we can investigate $3$ kinds learning rates in a  society:
%1.  $(\varepsilon, \bar\varepsilon, \{\delta_n\}_{n=1}^{\infty})$  2.$(\{\varepsilon_n\}_{n=1}^{\infty}, \bar\varepsilon, \delta)$   3. $(\varepsilon, \{\bar\varepsilon_n\}_{n=1}^{\infty}, \delta)$. Same as the previous subsection inquiring existence of perfect learning, we only work with the $(\varepsilon, \bar\varepsilon, \{\delta_n\}_{n=1}^{\infty})$ setup.

%In what follows we try to explore the achievable fastest learning rate. %namely, uniformly smallest $\{\delta_n\}_{n=1}^{\infty}$, for a given society $\{G_n\}_{n=1}^{\infty}$ with sequence of communication equilibria $\{\sigma^{n,*}\}_{n=1}^{\infty}$ that leads to perfect learning.
%By doing this we actually pin down the achievable highest perfect learning quality for the society $\{G_n\}_{n=1}^{\infty}$ with specific sequence of communication equilibria $\{\sigma^{n,*}\}_{n=1}^{\infty}$.
%With finite population learning as the foundation for perfect learning, we are able to discuss,  along the path of perfect learning,  achievable smallest $\delta_n$.  This consideration also helps us find the size of population needed to achieve a given tolerance.
%for finite population learning and perfect learning in a society of certain formation.

On the other hand,  it is not trivial to construct concretely the smallest sequence $\{\delta_n\}_{n=1}^{\infty}$ for perfect learning, while keeping other parameters fixed.  Recall that the sufficient condition part of Proposition \ref{fini} implies
$\delta$-perfect learning occurs with rates $\{\delta_n\}_{n=1}^{\infty}$ if
\begin{equation*} %\label{sufficientfinifoundation}
\frac{1}{n}\sum_{i=1}^n \erf\left(\varepsilon \sqrt{ \frac{\rho+\bar\rho k_i^{n,\sigma^*}}{2}}\right)\geqslant 1-\delta_n\bar\varepsilon\,.
\end{equation*}
Theoretically, if we can directly solve the inequalities with respect to $\delta_n$, the achievable fastest learning rate $\{\delta_n\}_{n=1}^{\infty}$ is constructed.  However, some technical problems prevent us from directly doing so.  First, we cannot read off a transparent rate out of the error function.  Second, without specific knowledge of network formations, the relation between $k_i^{n,\sigma^*}$ and $n$ is hard to be pinned down generally. Moreover, as we will see in the binomial tree example, $k_i^{n,\sigma^*}$ could be drastically different even for a same graph. 
Hence, we will first discuss learning rates on specific examples, and generalize to more general categories when possible.

\begin{examp}[\textsc{Isolated Agents}] \label{isolated}
When all agents are isolated from each other in a communication network $G_n$, we have $k_i^{n,\sigma^*}=1$ for every agent $i$.
\end{examp}
In Example \ref{isolated}, the negative condition (\ref{fininece}) is reduced to
$$
\erf\left(\varepsilon\sqrt{\frac{\rho+\bar\rho}{2}}\right) < (1-\bar\varepsilon)(1-\delta_n)\,.
$$
If parameters are such that $\erf\left(\varepsilon\sqrt{\frac{\rho+\bar\rho}{2}}\right) < (1-\bar\varepsilon)$, the above inequality holds for large $n$ for  any vanishing sequence $\{\delta_n\}_{n=1}^{\infty}$. This tells us that in fairly general circumstances, purely isolated society cannot achieve $\delta$-perfect learning.
%  In this case, learning rate is irrelevant.

\begin{examp}[\textsc{Complete Graph}] \label{complete}
When the communication network $G_n$ is a complete graph, and the benefit of getting $n-1$ new signals justifies the discount of one communication step, 
$k_i^{n,\sigma^*}=n$ for every agent $i$.
\end{examp}
In Example \ref{complete}, we have $$\erf\left(\varepsilon\sqrt{\frac{\rho+\bar\rho n}{2}}\right)\geqslant 1-\delta_n\bar\varepsilon\,, \quad\forall n\in\mathbb{N}\,,$$
as a sufficient condition for $\delta$-perfect learning, which translates to

%More generally, a sufficient condition for $(\varepsilon, \bar\varepsilon, \delta_n)$-learning is
%$$
%\delta_n\geq \frac{1}{\varepsilon}\left(1-\frac{1}{n}\sum_{i=1}^n \erf\left(\varepsilon \sqrt{\rho+\bar\rho \frac{k_i^{n,\sigma}}{2}}\right)\right)\,.
%$$
%
%Note that
%$$
%\frac{1}{\bar\varepsilon}\left(1-\frac{1}{n}\sum_{i=1}^n \erf\left(\varepsilon \sqrt{\rho+\bar\rho \frac{k_i^{n,\sigma}}{2}}\right)\right)\geq \frac{1}{\bar\varepsilon}\left(1-\frac{1}{n}\sum_{i=1}^n \erf\left(\varepsilon \sqrt{\rho+\bar\rho \frac{n}{2}}\right)\right)= \frac{1}{\bar\varepsilon}[1-\erf(\varepsilon \sqrt{\rho + \bar\rho n/2})]\,.
%$$
%Therefore, in this special case,  a sufficient condition for $(\varepsilon, \bar\varepsilon, \delta_n)$-learning on any communication graph $G_n$ and equilibrium strategy profile $\sigma$ is
\begin{equation} \label{nlower}
\delta_n \geqslant \frac{1}{\bar\varepsilon}\left(1-\erf\left(\varepsilon \sqrt{\frac{\rho + \bar\rho n}{2}}\right)\right)\,.
\end{equation}
The sequence of the right hand sides of inequality (\ref{nlower}) can serve as the learning rate.
At the cost of getting a conservative estimate, we approximate the error function in order to get a more transparent learning rate.
Note that the error function {\em erf} can be approximated by
$$
%\frac{1}{\sqrt{2\pi}}\frac{\sqrt{x}}{x+1}e^{-\frac{x}{2}}<
1-\erf(x) < \frac{1}{\sqrt{2\pi}}\frac{1}{x} e^{-x^2/2}\,.
$$
Therefore a sufficient condition for $\delta$-perfect learning is
$$
\delta_n\geqslant \frac{1}{\sqrt{\pi}\bar\varepsilon} \frac{1}{\varepsilon\sqrt{\rho+\bar\rho n}}\exp\left(-\frac{\varepsilon^2(\rho+\bar\rho n)}{4}\right)\,.
$$
Keep other parameters fixed, and focus on the relations between population size  $n$ and $\delta_n$. We see that $\delta_n$ could decrease in the order of $\exp\left(- \bar{\rho} \varepsilon^2 n/5 \right)$.
This implies that when population grows, the probability that at least $\bar\varepsilon$ fraction of people make the wrong decision decreases very quickly to zero. 
%In fact, it decreases faster than any polynomial rate.  This is not totally surprising as $k_i^{n,\sigma^*}=n$ for all agents.

Following the idea of error function approximations, we go beyond Example \ref{complete} to consider a more general case in which
%the $k_i^{n,\sigma^*}$'s of all agents in communication network $G_n$ with equilibrium $\sigma^{n,*}$ is lower bounded by a deterministic function $f(n)$.
%That is, suppose there exists a deterministic function $f$ such that
$k_i^{n,\sigma^*}\geqslant f(n)$ for every agent $i$ where $f(n)$ is a deterministic sequence.  A sufficient condition for $\delta$-perfect learning is then
\begin{equation} \label{minilower}
\delta_n\geqslant \frac{1}{\sqrt{\pi}\bar\varepsilon} \frac{1}{\varepsilon\sqrt{\rho+\bar\rho f(n)}}\exp\left(-\frac{\varepsilon^2(\rho+\bar\rho f(n))}{4}\right)\,.
\end{equation}
If $f(n)$ diverges to infinity as $n$ goes to infinity, the right hand side of inequality (\ref{minilower}) converges to $0$. Keeping other parameters fixed, this implies $\delta_n$ could decrease in the order of  $\exp(- \bar{\rho} \varepsilon^2 f(n)/5 )$.
Formally, we summarize these discussions with the next proposition.

\begin{prop} \label{fastestrate}
Suppose there exists a diverging sequence $f(n)$ such that $k_i^{n,\sigma^*}\geqslant f(n)$ for any agent $i$ in network $G_n$ with associated equilibrium $\sigma^{n,*}$, then
$\delta$-perfect learning could occur with learning rate $\{\delta_n\}_{n=1}^\infty$, where each $\delta_n$ is in the order of $\exp(- \bar{\rho} \varepsilon^2 f(n)/5 )$.
%\item[b)] $\delta$-perfect learning cannot occur if we insist on rate $\delta_n\sim$
%\end{itemize}
\end{prop}
%\textcolor{red}{The rate does not depend on $\bar\varepsilon$ and $\rho$.  it is quite intuitive that the rate should not depend on $\rho$, but not really clear about $\bar\varepsilon$.}
The next example is a direct application of Proposition \ref{fastestrate}.

\begin{examp}
Suppose $f(n)=C\cdot n$ where $0<C<1$, then $\delta$-perfect learning could occur with learning rate $\{\delta_n\}_{n=1}^\infty$, where each $\delta_n$ is in the order of $\exp(- \bar{\rho} \varepsilon^2 C n /5 )$.
\end{examp}
An interpretation of this example is that, even if communication is sparse in the sense that each of the agents can only get a small proportion of information in the entire population, perfect learning can still be reached at a fast rate.   This example represents a scenario in which communication networks in a society consist of dispersed social groups while agents within these social groups are very closely connected.  This is related to interesting results pertaining to social cliques or homophily as discussed in \cite{GJ2012b, GJ2012a, GJ2011}.

In most models, however, there is no universal bound for $k_i^{n,\sigma^*}$.  Lemma \ref{LEM: perfect1} in the supplementary materials renders Proposition \ref{fastestrate} as a special case, but it still does not cover cases when direct conservative estimate for $k_i^{n,\sigma^*}$ is not feasible.  Next we consider such a case: the binomial tree, which is widely considered as an axiomatic representation of various hierarchical structures in the human society \cite{J2010}.  
In particular, as the information flow within a binomial tree can be either from the root to the leafs or from the leafs to the root, binomial trees can accommodate both the top-down and the bottom-up cases of information transmission in various real-world scenarios.  Hence, it is instructive to analyze the binomial tree with a few different settings, where we generalize our game by allowing the information sensitiveness $\psi = \psi_n$ to vary along the society $\{G_n\}_{n=1}^{\infty}$.

\begin{examp}[\textsc{Binomial Tree: Information Flow from Root to Leafs}] \label{binomial}
The agents in the communication network $G_n$ form a binomial tree, where information can only flow from root to leafs. For simplicity, consider only the number of agents $n$ such that $n = 1 + 2 + 4+ ...+ 2^{(m_n-1)}$, where $m_n$ is the number of layers in the binomial tree. The following graph illustrates such a binomial tree with three layers.
\end{examp}
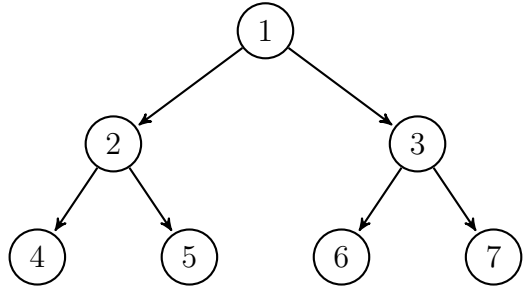
\begin{center}
\begin{tikzpicture}[->,>=stealth',shorten >=1pt,auto,node distance=2cm,
  thick,main node/.style={circle,fill=white!15,draw,font=\large\bfseries}]

  \node (1) at (3,3) [shape=circle, draw] {1};
  \node (2) at (1,1.5) [shape=circle, draw] {2};
  \node (3) at (5,1.5) [shape=circle, draw] {3};
  \node (4) at (0,0) [shape=circle, draw] {4};
  \node (5) at (2,0) [shape=circle, draw] {5};
  \node (6) at (4,0) [shape=circle, draw] {6};
  \node (7) at (6,0) [shape=circle, draw] {7};

  \path[every node/.style={font=\small}]
    (1) edge [left] (2)
    (1) edge [right] (3)
    (2) edge [right] (4)
    (2) edge [right] (5)
    (3) edge [right] (6)
    (3) edge [right] (7);
\end{tikzpicture}
\end{center}
We will study two scenarios of this binomial tree, in both of which $\lambda = r$ so that $\bar{r} = 1/2$.
%We stress the possible dependence of $\psi$ on $n$ and write it as $\psi_n$.

\begin{enumerate}
\item[i)] $\psi_n = \rho = \bar\rho = 1$.  For agent 1 on the top layer, he should exit right after step 0 because he does not have any chance to receive others' private information.  For agent 2 and 3, who are on the second top layer, they decide between step 0 and 1.  A simple calculation on their pay off functions reveals that they should exit after step 0. Agents 4, 5, 6, 7 who are on the third layer potentially should decide between 0,1 and 2 steps.  But since agents $2$ and $3$ cannot not pass through agent $1$'s info, step 2 is eliminated before any calculation.  So agents on the third layer actually faces same decision problems as agents on the second layer. Continue with the same argument till the $m_n$'th layer, we learn that everyone in the communication network exits right after she gets the private signal.  Therefore, this scenario is the same as \textit{isolated points} in terms of information aggregation.
\end{enumerate}
In general, as depicted in this subcase i), when the communication game is less information sensitive, the precision of the prior is higher, or the precision of the private signal is lower, it is less likely to achieve $\delta$-perfect learning, even if the agents are well connected.

\begin{enumerate}
\item[ii)] $\psi_n < \frac{2}{\rho + (m_n-1)\bar\rho} - \frac{1}{\rho + m_n \bar\rho} \text{ and } \varepsilon^2 < - \frac{4}{\bar\rho}\log \left(\frac{1}{2}\sqrt{\frac{\rho + 2\bar\rho}{\rho + \bar\rho}}\right)$.  Same as subcase i), agent 1 does not have a choice. For agents on the second layer to choose exit at step 1, we need $\psi_n < \frac{2}{\rho+ \bar\rho} - \frac{1}{\rho +  2\bar\rho}$.  For agents on the third layer to exit at step 2, we need
    $$
    \psi_n < \min\left\{\frac{2}{\rho + \bar\rho}-\frac{1}{\rho+ 2\bar\rho}, \frac{2}{\rho + 2\bar\rho}-\frac{1}{\rho+ 3\bar\rho}\right\} = \frac{2}{\rho + 2\bar\rho}-\frac{1}{\rho+ 3\bar\rho}\,.
    $$
    In general, an agent on layer $j$ wait till the $j-1$ step if
    $$
    \psi_n < \min \left\{ g(1), \ldots, g(j-1)  \right\} =  g(j-1)\,.
    $$
    where $g(x)= \frac{2}{\rho + x\bar\rho}- \frac{1}{\rho + (x+1)\bar\rho}$. The last equality holds because $g(x)$  is a decreasing function, thanks to $g'(x)< 0$.
    Hence under equilibrium, agents on layer $j$ have $j$ signals.  In particular, agents in the last layer each has $m_n = \log_2(n+1)$ signals. Note that there are $\frac{n+1}{2}$ agents in this layer.
    Using (3.2), a similar derivation to that in Example \ref{complete} leads to that the learning rate $\delta_n$ should be
    $$
    \delta_n\geqslant \frac{1}{n\varepsilon\bar\varepsilon\sqrt{\pi}}\sum_{i=1}^{\log_2(n+1)}2^{j-1}\frac{1}{\sqrt{\rho + \bar\rho j}}\exp\left(-\frac{\varepsilon^2(\rho+\bar\rho j)}{4}\right)\,.
    $$
    %$$
%    \delta_n \geq \frac{1}{\sqrt{\pi}\varepsilon\bar\varepsilon}\frac{1}{\sqrt{\rho + \bar\rho\log_2(n+1)}}\exp \left(-\frac{\varepsilon^2(\rho+\bar\rho\log_2(n+1))}{4}\right)\,.
%    $$
    To unravel right hand side of the above inequality, we let
    $$
    h(x) = 2^{x-1} \frac{1}{\sqrt{\rho + \bar\rho x}}\exp\left(-\frac{\varepsilon^2(\rho + \bar\rho x)}{4}\right)\,.
    $$
 Then $h(x)$ is monotone increasing, because  $h(x+1)/h(x)>1$ under our condition.  Therefore,
    it is sufficient to have
   $$
    \delta_n \geqslant \frac{1}{n\varepsilon\bar\varepsilon \sqrt{\pi}}\log_2(n+1) \cdot 2^{\log_2(n+1)-1}\frac{1}{\sqrt{\rho + \bar\rho \log_2(n+1)}}\exp\left(-\frac{\varepsilon^2(\rho + \bar\rho \log_2(n+1))}{4}\right)\,.
    $$
    Therefore $\delta_n$ could decay in the order of $\sqrt{\log(n+1)}\cdot (n+1)^{-\varepsilon^2\bar\rho/4}$, which is a polynomial rate.
\end{enumerate}

%This subcase ii) illustrates the scenario where $\delta$-perfect learning is reached.  This is achieved when the information exchange game is more information sensitive, the precision of the prior is lower, or the precision of the private signal is higher.  

Compared to the complete graph, the binomial tree aggregates information much slower.  The difference in learning rates arises not only from the physical network structures, but also from different strategic interactions among agents in the two environments.
Next, we consider a twin case of the binomial tree, in which information flows in the opposite direction.

\begin{examp}[\textsc{Binomial Tree: Information Flow from Leafs to Root}] \label{binomialReverse}
Now let information flow from leafs to root, i.e., reverse all the directed edges in Example \ref{binomial}. The following graph illustrates such a binomial tree with three layers.
\end{examp}
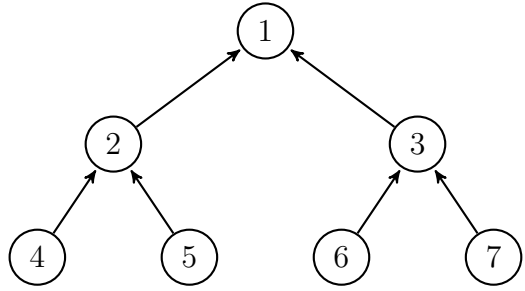
\begin{center}
\begin{tikzpicture}[->,>=stealth',shorten >=1pt,auto,node distance=2cm,
  thick,main node/.style={circle,fill=white!15,draw,font=\large\bfseries}]

  \node (1) at (3,3) [shape=circle, draw] {1};
  \node (2) at (1,1.5) [shape=circle, draw] {2};
  \node (3) at (5,1.5) [shape=circle, draw] {3};
  \node (4) at (0,0) [shape=circle, draw] {4};
  \node (5) at (2,0) [shape=circle, draw] {5};
  \node (6) at (4,0) [shape=circle, draw] {6};
  \node (7) at (6,0) [shape=circle, draw] {7};

  \path[every node/.style={font=\small}]
    (2) edge [left] (1)
    (3) edge [right] (1)
    (4) edge [right] (2)
    (5) edge [right] (2)
    (6) edge [right] (3)
    (7) edge [right] (3);
\end{tikzpicture}
\end{center}
We give the following results.  The detailed analysis is similar to Example \ref{binomial}.
\begin{itemize}
\item[i)] $\psi_n = \rho = \bar\rho = 1$. All agents exit after time $0$.
\item[ii)] $\psi_n < \frac{2}{\rho + 2^{(m_n -1)}\bar\rho} - \frac{1}{\rho + 2^{m_n}\bar\rho}$. All agents get the maximum number of signals that they could possibly get, then $\delta_n$ can be such that
$$
\delta_n \geqslant \frac{1}{n\varepsilon\bar\epsilon \sqrt{\pi}}\sum_{j=1}^{\log_2(n+1)} 2^{j-1} \frac{1}{\sqrt{\rho +2^{(m_n-j+1)}\bar\rho }}\exp\left(-\frac{\varepsilon^2(\rho +  2^{(m_n - j +1)}\bar\rho)}{4} \right)\,.
$$
A conservative estimate on the summation on the right hand side would give $\delta_n \sim n^{-3/4}$, a much faster rate than that in Example \ref{binomial} when $\varepsilon^2\bar\rho \ll 3$, which can be considered as a typical case as we have in mind very small $\varepsilon$.
\end{itemize}

Note that different information flow directions matter for learning rates. When parameters are in comparable range,  the bottom-up case exhibits a higher learning rate than the top-down case does.  In other words, the bottom-up organization of information flow within a binomial tree is more efficient in terms of aggregating information.  This result is consistent with early economics and sociology literature; \cite{H45} for example, highlight the importance of dispersed information sources.

%\textcolor{red}{Citation here: Hayek, F. A. 1945. The Use of Knowledge in Society,
%The American Economic Review , Vol. 35, No. 4, pp. 519-530}

The four sub-cases under binomial tree setting demonstrate that beyond a directed graphical model, contextual information is also very important for information aggregation.  These comparisons are made possible only with help of our concept of finite population learning and $\delta$-perfect learning.  Properties or statistics of the graph alone cannot determine the learning status.  Rather, as we argued from the very beginning of this work, the enriched game theory plus graphical modeling approach are both interesting and necessary in helping understand the information aggregation in social networks.

\section{Remarks and Further Research}

We have proposed a finite population learning concept to capture the level of information aggregation in any given communication network. In our framework, one equilibrium outcome, i.e., the number of signals obtained by an agent when she makes a decision, plays a key role. This equilibrium outcome is computable (\cite{MM1996}), which also allows us to numerically check the learning status. Different from existing literature that mainly addresses the learning behavior at the limit, this new concept helps reveal explicit interplays among time discount, frequency of communication, information precision and information sensitiveness of the decision problem in any finite communication network. As the total amount of information is fixed in a given finite network, our approach enables meaningful comparative statics regarding the effectiveness of information aggregation in networks.  We also provide conditions for learning under a particular equilibrium, under any equilibrium, and under all equilibria, respectively. Thanks to its tractability and transparency, the finite population learning concept offers a solid foundation to investigate long run dynamics of learning behavior and the associated learning rates as population diverges.

Our analysis is also subject to certain limitations, which suggest directions for future research. In our model, complete information on the structure of a given communication network is required in determining the number of signals obtained by an agent, and in checking its corresponding learning status. In some circumstances, researchers do not want to assume such specific information; rather they want to get some understanding of the learning status regarding a large class of networks. This goal calls for some new criteria that can determine the learning status for given classes of networks with given finite population; ideally, these criteria should be formulated in terms of some summary statistics of these networks.  Relaxing the knowledge on specific network structure may also lead to more general results about the learning rates. However, this task is challenging within the current finite population learning framework.  Specifically, our established conditions for finite population learning require all agents' exact numbers of signals upon their exits. Only knowing some commonly used summary statistics of the associated graphs can hardly help offer good estimates of these numbers of signals, mostly because these numbers of signals are also affected by other parameters not directly related to the network structure, such as the information precisions and the tolerances of learning. Therefore, the learning status of a certain class of communication networks is largely undetermined if we just consider properties of the graphs. To address this issue, we would like to have novel statistical properties of communication networks that are more friendly to the analysis of communication learning.  \cite{GJ2012b, GJ2012a, GJ2011} are promising attempts towards this direction, but their results are limited to the context of non-Bayesian observational learning.

Another line of generalization is to make our model more flexible and realistic.   For example, the current setting assumes that agents have private signals with the same precision, which amounts to assuming that the total amount of information grows linearly with the population size when we consider the perfect learning.  It might be interesting to allow the total information to increase in a nonlinear (e.g., log rate) fashion with the population size, and allow non-uniform distribution of signal precisions among agents.  Also, even when we focus on certain classes of networks without specifying complete network structure, it is still assumed that any agent in the communication network knows the complete network structure. This assumption can be relaxed by limiting agents' knowledge on the network to a certain local neighborhood, and infer other parts of the network according to her local knowledge.     Keeping the Bayesian learning paradigm, other potential generalizations of our current work include considering the implications of correlated private information among agents, and heterogeneous characteristics of agents.

%%\newpage
%
%%\begin{thebibliography}{99}
%
%
%%\newpage
%
%
%
\newpage
\setcounter{page}{1}
\section*{Supplementary Materials}
\setcounter{equation}{0}

The supplementary materials provide proofs and generalized results of corresponding parts in the main text.

\begin{proof}[\textsc{Proof of Proposition \ref{fini}.}]
To prevent $(\varepsilon, \bar\varepsilon, \delta)$-learning, it is enough to show that a lower bound of $\mathbb{P}_{\sigma^{n,*}}\left(\frac{1}{n}\sum_{i=1}^n \left(1-M_i^{n,\varepsilon}\right)\geqslant \bar\varepsilon\right)$ is greater than $\delta$.
It follows from Markov inequality,
\begin{equation*}
\mathbb{P}_{\sigma^{n,*}}\left(\frac{1}{n}\sum_{i=1}^n M_i^{n,\varepsilon} > 1-\bar\varepsilon\right)
\leqslant 
n^{-1}(1-\bar\varepsilon)^{-1}\sum_{i=1}^n \mathbb{E}_{\sigma^{n,*}}M_i^{n,\varepsilon}
=n^{-1}(1-\bar\varepsilon)^{-1}\sum_{i=1}^n \erf\left(\varepsilon\sqrt{\frac{\rho+\bar\rho k_i^{n,\sigma^*}}{2}}\right)\,.
\end{equation*}
This implies that
\begin{equation*}
\mathbb{P}_{\sigma^{n,*}}\left(\frac{1}{n}\sum_{i=1}^n\left(1-M_i^{n,\varepsilon}
\right)\geqslant\bar\varepsilon\right) %=\mathbb{P}_{\sigma^{n,*}}\left(\frac{1}{n}\sum_{i=1}^n M_i^{n,\varepsilon}\leqslant 1-\bar\varepsilon\right)
>1-n^{-1}(1-\bar\varepsilon)^{-1}\sum_{i=1}^n \erf\left(\varepsilon\sqrt{\frac{\rho+\bar\rho k_i^{n,\sigma^*}}{2}}\right)\,.
\end{equation*}
Therefore, it is enough to take
\begin{equation*}
1-n^{-1}(1-\bar\varepsilon)^{-1}\sum_{i=1}^n \erf\left(\varepsilon\sqrt{\frac{\rho+\bar\rho k_i^{n,\sigma^*}}{2}}\right) > \delta\,,
\end{equation*}
which concludes that condition (\ref{fininece}) is necessary for $(\varepsilon, \bar\varepsilon, \delta)$-learning.

To ensure $(\varepsilon, \bar\varepsilon, \delta)$-learning, note that
\begin{equation*}
\mathbb{P}_{\sigma^{n,*}}\left(\frac{1}{n}\sum_{i=1}^n \left(1-M_i^{n,\varepsilon}\right)\geqslant \bar\varepsilon\right)\leqslant \frac{\mathbb{E}_{\sigma^{n,*}}\left(\sum_{i=1}^n(1-M_i^{n,\varepsilon})\right)}{n\bar\varepsilon}= \frac{n-\sum_{i=1}^n \erf\left(\varepsilon\sqrt{\frac{\rho+\bar\rho  k_i^{n,\sigma^*}}{2}}\right)}{n\bar\varepsilon}\,.
\end{equation*}
Demanding the right hand side of the above inequality no larger than $\delta$,  is the same as assuming condition (\ref{finisuff}).  This completes the proof.
\end{proof}

The following provides a more general sufficient condition for $\delta$-perfect learning. Given equilibria $\{\sigma^{n,*}\}_{n=1}^{\infty}$,
let $f_1\geqslant f_2\geqslant\ldots\geqslant f_J$, where each $f_j(n)$ is a monotone increasing function (not necessarily strictly increasing) on $n$, and let $\{b_n^j$, $j=1, \ldots, J\}$ be such that
$$
\frac{|\{i:k_i^{n,\sigma^*}\geqslant f_1(n)\}|}{n}\geqslant b_n^1\,,
$$
$$
\frac{|\{i:f_1(n)> k_i^{n,\sigma^*}\geqslant f_2(n)\}|}{n}\geqslant b_n^2\,,
$$
and up until
$$
\frac{|\{i:f_{J-1}(n)> k_i^{n,\sigma^*}\geqslant f_J(n)\}|}{n}\geqslant b_n^J\,.
$$
Clearly, $b_n^1, \ldots, b_n^J\in(0,1)$ and $0\leqslant b_n^1+\ldots+ b_n^J\leqslant 1$.   The rest agents $i$'s are such that $f_J(n) > k_i^{n,\sigma^*}\geqslant 1$.  Their fraction is at most $1-(b_n^1+\ldots+b_n^J)$.

\begin{lem}\label{LEM: perfect1}
$\delta$-perfect learning occurs if
\begin{itemize}
\item[(a)] $\lim_{n\to\infty}\sum_{j=1}^J b_n^j=1$,
\item[(b)] for each $j\in\{1\ldots, J\}$,
$
\lim_{n\to\infty}b_n^j\left(1-\erf\left(\varepsilon\sqrt{ \frac{\rho + \bar\rho f_j(n)}{2}}\right)\right)=0\,.
$
\end{itemize}
\end{lem}

\begin{proof}[\textsc{Proof of Lemma \ref{LEM: perfect1}.}]
Recall that a sufficient condition for $(\varepsilon, \bar\varepsilon, \delta_n)$-learning is
$$
\frac{1}{n}\sum_{i=1}^n \erf\left(\varepsilon\sqrt{ \frac{\rho +\bar\rho k_i^{n,\sigma^*}}{2}}\right)\geqslant 1-\delta_n\varepsilon\,.
$$
Then by the definition of $b_n^j$ and $f_j$, it is enough to have
$$
\sum_{j=1}^J \erf\left(\varepsilon \sqrt{\frac{\rho+\bar\rho  f_j(n)}{2}}\right)\cdot b_n^j+ \left(1-\sum_{j=1}^J b^j_n\right)\erf\left(\varepsilon\sqrt{\frac{\rho+\bar\rho}{2}}\right)\geqslant 1-\delta_n\varepsilon\,.
$$
This translates to
$$
\delta_n\geqslant \frac{1}{\bar\varepsilon}\left(\sum_{j=1}^J b_n^j\left(1-\erf\left(\varepsilon\sqrt{\frac{\rho+\bar\rho  f_j(n)}{2}}\right)\right)+\left(1-\sum_{j=1}^J b_n^j\right)\left(1-\erf\left(\varepsilon\sqrt{\frac{\rho+\bar\rho}{2}}\right)\right)\right)\,.
$$
To ensure the existence of $\{\delta_n\}$ such that $\lim_{n\to\infty}\delta_n= 0$, it is enough to have $\lim_{n\to\infty}\sum_{j=1}^J b_n^{j} = 1$ and
$$
\lim_{n\to\infty}b_n^j\left(1-\erf\left(\varepsilon\sqrt{ \frac{\rho + \bar\rho f_j(n)}{2}}\right)\right)= 0, \quad \mbox{ for } j \leq J.
$$

\end{proof}

Note that if $f_j$ does not increase strictly for $n\geqslant N^*$, $b_n^j$ needs to decrease to $0$.
Also, allowing more than one tolerances among $\varepsilon$, $\bar\varepsilon$, $\delta$ to vary with population size $n$ leads to interesting learning results.  In particular,
from the proof of Lemma  \ref{LEM: perfect1},  a sufficient condition for $(\varepsilon, \bar\varepsilon_n, \delta_n)$- learning is
$$
\delta_n \bar{\varepsilon}_n\geqslant \sum_{j=1}^J b_n^j\left(1-\erf\left(\varepsilon\sqrt{\frac{\rho+\bar\rho f_j(n)}{2}}\right)\right)+\left(1-\sum_{j=1}^J b_n^j\right)\left(1-\erf\left(\varepsilon\sqrt{\frac{\rho+\bar\rho}{2}}\right)\right)\,.
$$
In this condition, the role of $\delta_n$ and that of $\bar\varepsilon_n$ are completely interchangeable, which implies that we can trade in some probabilistic confidence for some fraction of agents who make wrong decisions.

The following provides a generalized version of Lemma \ref{round short} when Assumption \ref{Assumption:regular} is relaxed.

\begin{lem}[\textsc{Generalized Lemma \ref{round short}}] \label{round}
For any agent $i$, either the communication steps she optimally experiences before taking an action in any social network $G_n$ along a society $\{G_n\}_{n=1}^\infty$ is bounded from above by a constant independent of $n$, or she waits until the maximum allowable step.  Specifically,
\begin{itemize}
\item[(a)] If $(\rho+\bar{\rho})\psi>1$, then for any agent $i$
\begin{equation*}
l_i^{n, \sigma^*}\leqslant l_i^n<\min\left\{(L_i^n)_{max}, \ln\left(1-\frac{1}{(\rho+\bar{\rho})\psi}\right)/\ln\bar{r}\right\}\,,
\end{equation*}
where $l_i^n$ stands for agent i's optimal communication rounds given that other agents wait till the maximum allowable step.
\item[(b)] If $(\rho+\bar{\rho})\psi\leqslant 0$ (equivalently, $\psi \leq 0$), then for any agent $i$
\begin{equation*}
l_i^{n, \sigma^*}=l_i^n=(L_i^n)_{max}\,.
\end{equation*}
\item[(c)] If $0<(\rho+\bar{\rho})\psi \leqslant 1$, then there are two subcases.
\begin{itemize}
\item[(c.1)] For agent $i$ with
\begin{equation*}
\lim_{n\to\infty}|B_i^n|<\frac{1-\rho\psi}{\bar{\rho}\psi}\,,
\end{equation*}
where $B_i^n$ is the set of agents whose signals agent $i$ can get if no one exits before maximum allowable step, we have
\begin{equation*}
l_i^{n, \sigma^*}=l_i^n=(L_i^n)_{max}\,.
\end{equation*}
\item[(c.2)] For agent $i$ with
\begin{equation*}
\lim_{n\to\infty}|B_i^n|\geqslant \frac{1-\rho\psi}{\bar{\rho}\psi}\,,
\end{equation*}
we have either
\begin{equation*}
l_i^{n, \sigma^*}\leqslant l_i^n\leqslant \min\left((L_i^n)_{max}, l^{\{G_n\}_{n=1}^\infty}_i\right)\,,
\end{equation*}
or
\begin{equation*}
l_i^{n, \sigma^*}=(L_i^n)_{max}\,,  
\end{equation*}
%or
%\begin{equation*}
%l_i^{n, \sigma^*}=l_i^n=(L_i^n)_{max}\,,
%\end{equation*}
where $l^{\{G_n\}_{n=1}^\infty}_i$ is a constant that depends on the society and agent $i$'s position in the society, but does not change with $n$.
\end{itemize}
\end{itemize}
\end{lem}

\begin{comment}
(c.2). For agent $i$ with
\begin{equation*}
\frac{1-\rho\psi}{\bar{\rho}\psi}\leqslant \lim_{n\to\infty}|B_i^n|<\infty \,,
\end{equation*}
we have either
\begin{equation*}
S_i^{n, \sigma}\leqslant S_i^n<S^{\{G_n\}_{n=1}^\infty}_i\,.
\end{equation*}
or
\begin{equation*}
S_i^{n, \sigma}=\infty,  S_i^n<S^{\{G_n\}_{n=1}^\infty}_i\,.
\end{equation*}
(c.3). For agent $i$ with
\begin{equation*}
\lim_{n\to\infty}|B_i^n|=\infty\,,
\end{equation*}
we have
\begin{equation*}
S_i^{n, \sigma}= S_i^n<S^{\{G_n\}_{n=1}^\infty}_i\,.
\end{equation*}
In case (c.2) and (c.3), $S^{\{G_n\}_{n=1}^\infty}_i$ is a constant that depends on the society and agent $i$'s position in the society, and it does not change with respect to $n$.
\end{comment}

\begin{comment}
\begin{equation*}
\left\{
\begin{array}{ll}
\ln\left[1-\frac{1}{(\rho+\bar{\rho})\psi}\right]/\ln\bar{r}\,, & \textrm{ if  } (\rho+\bar{\rho})\psi>1\,,\\
S^{\{G_n\}_{n=1}^\infty}, & \textrm{ if } 0<(\rho+\bar{\rho})\psi\leqslant 1\,.
\end{array}\right.
\end{equation*}
\end{comment}

\begin{proof}[\textsc{Proof of Lemma \ref{round}.}]
We proceed case by case.\\
\indent \textsc{Case }(a), $(\rho+\bar{\rho})\psi>1$.

In this case, agent $i$ enjoys a positive payoff $\psi-\frac{1}{\rho+\bar{\rho}}$ if she exists at $t=0$ and does not communicate with anyone else.  Note that her expected payoff  by taking $l_i^n$ communication steps is strictly upper bounded by $\bar{r}^{l_i^n} \psi$.  Therefore, it is suboptimal for her to choose a $l_i^n$ such that
\begin{equation*}
\bar{r}^{l_i^n} \psi\leqslant \psi-\frac{1}{\rho+\bar{\rho}}\,,
\end{equation*}
which implies
\begin{equation*}
l_i^n<\ln\left(1-\frac{1}{(\rho+\bar{\rho})\psi}\right)/\ln\bar{r}
\end{equation*}
is necessary for agent $i$'s optimality.  It is obvious that $l_i^{n, \sigma^*}\leqslant l_i^n$, since other agents do not necessarily wait forever in an equilibrium, so that it may be optimal for agent $i$ to exit earlier too. We get the result by combining these with the upper bound $l_i^n\leqslant (L_i^n)_{max}$.  
\quad\\
\indent \textsc{Case }(b), $(\rho+\bar{\rho})\psi\leqslant 0$.

Now agent $i$ always gets a negative payoff whenever she exits.  Because waiting discounts the negative payoff, she optimally chooses to wait as long as possible, no matter what other agents do.  Therefore, $l_i^{n, \sigma^*}=l_i^n=(L_i^n)_{max}$.

\indent \textsc{Case }(c.1), $0<(\rho+\bar{\rho})\psi\leqslant 1$ and $\lim_{n\to\infty}|B_i^n|<\frac{1-\rho\psi}{\bar{\rho}\psi}$.

The maximum number of private signals agent $i$ can get is $|B_i^n|$. Again, agent $i$ always gets a negative payoff whenever she exists.  Hence, $l_i^{n, \sigma^*}=l_i^n=(L_i^n)_{max}\,.$

\indent \textsc{Case }(c.2), $0<(\rho+\bar{\rho})\psi\leqslant 1$ and $\lim_{n\to\infty}|B_i^n|\geqslant \frac{1-\rho\psi}{\bar{\rho}\psi}$.

%On the one hand, for any $G_n$ with $|B_i^n|<\frac{1-\rho\psi}{\bar{\rho}\psi}$, agent $i$ always gets a negative payoff whenever she exists.  Hence, this again corresponds to case (b) and $l_i^{n, \sigma^*}=l_i^n=(L_i^n)_{max}\,.$

For any $G_n$ with $|B_i^n|\geqslant \frac{1-\rho\psi}{\bar{\rho}\psi}$, we consider the communication step $(L_i^n)_{max}$ when agent $i$ obtains signals from all her sources $B_i^n$, provided others wait maximum steps. 
Note that $(L_i^n)_{max}$ is non-decreasing in $n$ for any agent $i$ (by the no deleting assumption), and $|B_{i,l}^n|$ is strictly monotone increasing in $l$ when $l\leqslant (L_i^n)_{max}$.

Also for a given communication network $G_n$, there exists one communication step $l_i^{n'}$ such that after this step agent $i$ gets positive payoff, given that other agents wait maximum steps.  Hence, it is suboptimal for her to wait longer than $l_i^{n'}$ if
\begin{equation*}
\bar{r} \psi \leqslant \psi - \frac{1}{\rho+\bar{\rho}|B_{i,l_i^{n'}}^n|}\,,
\end{equation*}
which implies
\begin{equation} \label{discountbound}
|B_{i,l_i^n}^n|<\frac{\lambda+r-\rho r\psi}{\bar{\rho}r\psi}
\end{equation}
is necessary for agent $l_i^n$'s optimality.

Now we consider two sub-cases.  First is when $\lim_{n\to\infty}|B_i^n|<\infty$.  Then we must have
$\lim_{n\to\infty}(L_i^n)_{max}<\infty\,,$ since $(L_i^n)_{max}\leqslant|B_i^n|$.  %Recall that $(L_i^n)_{max}$ is non-decreasing in $n$ for any $i$, we further have $(L_i^n)_{max}\leqslant \lim_{n\to\infty}(L_i^n)_{max}\,.$ 
Hence, 
\begin{equation*}
l_i^n\leqslant \lim_{n\to\infty}(L_i^n)_{max}<\infty\,,
\end{equation*}
for all $G_n$ satisfying $|B_i^n|\geqslant \frac{1-\rho\psi}{\bar{\rho}\psi}$ and $\lim_{n\to\infty}|B_i^n|<\infty$.  We denote $\lim_{n\to\infty}(L_i^n)_{\max}$ as $l^{\{G_n\}_{n=1}^\infty}_{1i}$, which is a constant that depends on the society and agent $i$'s position in the society and does not change with respect to $n$.

%\textcolor{red}{$n_0$}

Second, we discuss the case when $\lim_{n\to\infty}|B_i^n|=\infty$.  Now there should be either $\lim_{n\to\infty}(L_i^n)_{max}<\infty$ or $\lim_{n\to\infty}(L_i^n)_{max}=\infty\,.$  In the former scenario,  we have $l_i^n\leqslant \lim_{n\to\infty}(L_i^n)_{max} = l^{\{G_n\}_{n=1}^\infty}_{1i}$ for all $G_n$.  In the latter case, as $(L_i^n)_{max}$ is non-decreasing in $n$ for any given $i$ and $|B_{i,l}^n|$ is strictly monotone increasing in $l$ when $l\leqslant (L_i^n)_{max}$ for any $G_n$, there exists a largest $G_N$ with its associated $(L_i^N)_{max}$ that satisfies condition  (\ref{discountbound}).  Hence, by  (\ref{discountbound}) we obtain
\begin{equation*}
l_i^n\leqslant (L_i^N)_{max}\,,
\end{equation*}
%%\textcolor{red}{something more}
for all $G_n$ satisfying $|B_i^n|\geqslant \frac{1-\rho\psi}{\bar{\rho}\psi}$, $\lim_{n\to\infty}|B_i^n|=\infty$ and $\lim_{n\to\infty}(L_i^n)_{max}=\infty\,.$  We denote such $(L_i^N)_{max}$ as $l^{\{G_n\}_{n=1}^\infty}_{2i}$, which is again a constant that depends on the society and agent $i$'s position in the society and does not change with respect to $n$.
To sum up, we denote by $l^{\{G_n\}_{n=1}^\infty}_{i}$ either $l^{\{G_n\}_{n=1}^\infty}_{1i}$ or $l^{\{G_n\}_{n=1}^\infty}_{2i}$ in respective cases,  and it follows $l_i^n\leqslant l^{\{G_n\}_{n=1}^\infty}_i$ for agent $i$ in such $G_n$ with $|B_i^n|\geqslant \frac{1-\rho\psi}{\bar{\rho}\psi}$, where $l^{\{G_n\}_{n=1}^\infty}_i$ is independent of $n$.

As for $l_i^{n, \sigma^*}$, since other agents play equilibrium strategies, agent $i$ gets weakly fewer signals than that she can get when other agents wait maximum steps. There can be two cases, either she gets positive payoff and takes an action weakly earlier, namely, $l_i^{n, \sigma^*}\leqslant l_i^n$, or she cannot get enough signals to ensure a positive payoff so that she optimally until the maximum allowable step, i.e., $l_i^{n, \sigma^*}=(L_i^n)_{max}$.  This concludes the proof.
\end{proof}

\begin{proof}[\textsc{Proof of Proposition \ref{socialinformedlearn}.}]
By Lemma \ref{LEM: perfect1}, it suffices to show that $\lim_{n\to\infty}k_i^{n,\sigma^*}=\infty$ under any equilibria $\{\sigma^{n,*}\}_{n=1}^{\infty}$
for any socially informed agent $i$.  In the following, we consider a fixed socially informed agent $i$.  Recall that in Definition \ref{socialinformed}, $L_i$ is defined as the smallest positive integer such that $\lim_{n\to\infty} |B_{i,L_i}^n| = \infty\,.$ Denote by $B_{i, l}^{n,\sigma^*}$ the set of agents whose signals can reach $i$ in the first $l$ rounds of communication under equilibrium $\sigma^{n,*}$.

The problem is simple when $L_i=1$.  Clearly, $B_{i, 1}^{n,\sigma^*}=B_{i, 1}^{n}$ under any equilibrium $\sigma^{n,*}$.  As agent $i$ is socially informed, we have for sufficiently large $n$
\begin{equation*}
\psi-\frac{1}{\rho+\bar{\rho}|B_{i,1}^{n,\sigma^*}|}>0 \text{ under any } \sigma^{n,*}
\end{equation*}
and
\begin{equation*}
\bar{r}\left(\psi-\frac{1}{\rho+\bar{\rho}|B_{i,1}^{n,\sigma^*}|}\right) > \psi-\frac{1}{\rho+\bar{\rho}} \text{ under any } \sigma^{n,*}\,.
\end{equation*}
The above display implies that agent $i$ should at least wait for one communication round.  Hence, $k_i^{n,\sigma^*}\geqslant  |B_{i, 1}^{n,\sigma^*}|$ under any $\sigma^{n,*}$ for sufficiently large $n$.  As a consequence, $\lim_{n\to\infty}k_i^{n,\sigma^*}\geqslant \lim_{n\rightarrow\infty} |B_{i, 1}^{n,\sigma^*}|=\lim_{n\rightarrow\infty}|B^n_{i, 1}|=\infty$ under any $\{\sigma^{n,*}\}_{n=1}^{\infty}$.

The following discussion is on  the cases when $L_i\geqslant 2$.  We proceed through three steps.

\textsc{Step 1.} We claim when $L_i\geqslant 2$, for sufficiently large n, there exists at least one path $\{j_{L_i-1}, j_{L_i-2}, ..., j_1, i\}$ from $j_{L_i-1}$ to $i$ such that
\begin{equation} \label{informpath}
\lim_{n\to\infty}|B_{j_{L_i-l}, l}^n|=\infty \text{ for all } l\in\{1, \ldots, L_i-1\}\,.
\end{equation}
Now we construct the path $\{j_{L_i-1}, j_{L_i-2}, ..., j_1, i\}$ that satisfies condition (\ref{informpath}). Because $L_i$ is the smallest integer $j$ such that $\lim_{n\to\infty}|B^n_{i,j}|=\infty$, $B^n_{i, L_i-1} \setminus B^n_{i, L_i-2}$, the set of agents that are of distance $L_i-1$ to $i$, must be finite in the limit, i.e., $\lim_{n\to\infty}|B^n_{i, L_i-1} \setminus B^n_{i, L_i-2}|<\infty$.  Therefore, there is at least one agent $j$ of distance $L_i-1$ to $i$, such that $\lim_{n\to\infty}|B^n_{j,1}|=\infty$.
%By Definition \ref{socialinformed}, we know that $\lim_{n\to\infty}|B^n_{i, S_i-1} \setminus B^n_{i, S_i-2}|<\infty\,,$ so that there must exist at least one agent $j$ of distance $S_{i-1}$ to $i$ that satisfies $\lim_{n\to\infty}|B^n_{j,1}|=\infty\,;$ otherwise we have $\lim_{n\to\infty} |B^n_{i,S_i}| <\infty\,,$ which is a contradiction.
We denote one of such agents $j$ as $j_{L_i-1}$.  If $L_i=2$, the desired path has been constructed.  When $L_i\geqslant 3$, choose any path $\{j_{L_i-1}, j_{L_i-2}, ..., j_1, i\}$ from the chosen $j_{L_i-1}$ to $i$.  Clearly, $j_{L_i-l}\in B^n_{i, L_i-l}$.  Moreover, condition (\ref{informpath}) is satisfied in view of $\lim_{n\to\infty}|B^n_{j_{L_i-1},1}|=\infty$.

\textsc{Step 2.} We next argue that when $L_i\geqslant 2$, agent $j_{L_i-l}$ on the path $\{j_{L_i-1}, j_{L_i-2}, ..., j_1, i\}$ will not exit before she experiences $l$ communication steps under any equilibrium $\sigma^{n,*}$ provided that $n$ is sufficiently large.  It is worth noting that agent $j_{L_i-l}$ does not necessarily get a positive payoff when she experiences $l$ communication steps in equilibrium.

We will see this by induction from $j_{L_i-1}$ to $j_1$ sequentially.  We first show that agent $j_{L_i-1}$ will not exit before she experiences her first communication step in any equilibrium $\sigma^{n,*}$ provided that $n$ is sufficiently large.  It requires that there exists $N$ such that for all social networks $G_n\in \{G_n\}_{n=1}^\infty$ and its associated equilibrium $\sigma^{n,*}$ with $n\geqslant N$,
\begin{equation} \label{proofsi1}
\bar{r}\left(\psi-\frac{1}{\rho+\bar{\rho}|B_{j_{L_i-1}, 1}^{n,\sigma^*}|}\right) > \psi-\frac{1}{\rho+\bar{\rho}}\,.
\end{equation}

To validate condition (\ref{proofsi1}), recall condition (\ref{incentiveinform}) from Definition \ref{socialinformed} for $l=L_i-1$, which states that  there exists $N$ such that for all social networks $G_n\in \{G_n\}_{n=1}^\infty$ with $n\geqslant N$ it holds
\begin{equation} \label{proofsi2}
\bar{r}\left(\psi-\frac{1}{\rho+\bar{\rho}|B_{i,L_i}^n|}\right)> \psi-\frac{1}{\rho+\bar{\rho}|B_{i,L_i-1}^n|}\,.
\end{equation}
By the definition of $L_i$, the construction of $j_{L_i-1}$ and the fact that $B_{j_{L_i-1}, 1}^{n,\sigma^*}=B_{j_{L_i-1}, 1}^n$ under any equilibrium $\sigma^{n,*}$ with any $n$, we know that $\lim_{n\to\infty}|B_{j_{L_i-1}, 1}^{n,\sigma^*}|=\lim_{n\to\infty}|B_{j_{L_i-1}, 1}^{n}|=\infty$ under any $\sigma^{n,*}$ and $\lim_{n\to\infty}|B_{i,L_i-1}^n|<\infty$.  Also we have $|B_{i,L_i-1}^n|\geqslant 1$.  Note that the right hand side of condition (\ref{proofsi2}) is greater than or equal to the right hand side of condition (\ref{proofsi1}),  we obtain easily that (\ref{proofsi1}) holds for sufficiently large $n$.  Hence we get that agent $j_{L_i-1}$ will not exit before she experiences her first communication step under any $\sigma^{n,*}$ provided that $n$ is sufficiently large.

We then show that agent $j_{L_i-2}$ (for $L_i\geq 3$) will not exit before she experiences her second communication step under any equilibrium for sufficiently large $n$.  It requires that there exists $N$ such that for all social networks $G_n\in \{G_n\}_{n=1}^\infty$ and its associated equilibrium $\sigma^{n,*}$ with $n\geqslant N$,
\begin{equation} \label{proofsi3}
\bar{r}^2\left(\psi-\frac{1}{\rho+\bar{\rho}|B_{j_{L_i-2}, 2}^{n,\sigma^*}|}\right)> \psi-\frac{1}{\rho+\bar{\rho}}\,,
\end{equation}
and
\begin{equation} \label{proofsi4}
\bar{r}^2\left(\psi-\frac{1}{\rho+\bar{\rho}|B_{j_{L_i-2}, 2}^{n,\sigma^*}|}\right) > \bar{r}\left(\psi-\frac{1}{\rho+\bar{\rho}|B_{j_{L_i-2}, 1}^{n,\sigma^*}|}\right) \,.
\end{equation}

To validate (\ref{proofsi3}) and  (\ref{proofsi4}), we use again the condition (\ref{incentiveinform}) from Definition \ref{socialinformed} for $l=L_i-2$ and $l= L_i-1$, which state that there exists $N$ such that for all social networks $G_n\in \{G_n\}_{n=1}^\infty$ with $n\geqslant N$ we have
\begin{equation} \label{proofsi5}
\bar{r}^2\left(\psi-\frac{1}{\rho+\bar{\rho}|B_{i,L_i}^n|}\right) > \psi-\frac{1}{\rho+\bar{\rho}|B_{i,L_i-2}^n|}\,,
\end{equation}
and
\begin{equation} \label{proofsi6}
\bar{r}^2\left(\psi-\frac{1}{\rho+\bar{\rho}|B_{i,L_i}^n|}\right) > \bar{r}\left(\psi-\frac{1}{\rho+\bar{\rho}|B_{i,L_i-1}^n|}\right)\,.
\end{equation}

Similarly, by the definition of $L_i$ and the construction of $j_{L_i-1}$ and $j_{L_i-2}$, we know that $\lim_{n\to\infty}|B_{j_{L_i-2}, 2}^n|=\lim_{n\to\infty}|B_{i,L_i}^n|=\infty$,  $\lim_{n\to\infty}|B_{i,L_i-1}^n|<\infty$, $\lim_{n\to\infty}|B_{i,L_i-2}^n|<\infty$, and $\lim_{n\to\infty}|B_{j_{L_i-2},1}^{n,\sigma^*}|\leqslant\lim_{n\to\infty}|B_{j_{L_i-2},1}^n|<\infty$ under any equilibrium $\sigma^{n,*}$.  Also we have $|B_{i,L_i-2}^n|\geqslant 1$ and $B_{j_{L_i-2},1}^{n,\sigma^*}\subseteq B_{j_{L_i-2},1}^n\subseteq B_{i,L_i-1}^n$ (and thus $| B_{i,L_i-1}^n|\geqslant |B_{j_{L_i-2},1}^n|\geqslant |B_{j_{L_i-2},1}^{n,\sigma^*}|$) for any $n$ under any equilibrium $\sigma^{n,*}$.   Note that the right hand side of condition (\ref{proofsi5}) is greater than or equal to the right hand side of condition (\ref{proofsi3}), and the right hand side of condition (\ref{proofsi6}) is greater than or equal to the right hand side of condition (\ref{proofsi4}). Then it can be verified that the next two inequalities hold for sufficiently large $n$, the right hand sides of which are the same as those in conditions (\ref{proofsi3}) and (\ref{proofsi4}):
\begin{equation} \label{proofsi7}
\bar{r}^2\left(\psi-\frac{1}{\rho+\bar{\rho}|B_{j_{L_i-2}, 2}^{n}|}\right)> \psi-\frac{1}{\rho+\bar{\rho}}\,,
\end{equation}
and
\begin{equation} \label{proofsi8}
\bar{r}^2\left(\psi-\frac{1}{\rho+\bar{\rho}|B_{j_{L_i-2}, 2}^{n}|}\right) > \bar{r}\left(\psi-\frac{1}{\rho+\bar{\rho}|B_{j_{L_i-2}, 1}^{n, \sigma^*}|}\right) \,.
\end{equation}

Furthermore, recall that we have already shown that agent $j_{L_i-1}$ will not exit before she experiences her first communication step under any equilibrium $\sigma^{n,*}$ provided that $n$ is sufficiently large, which implies that $B_{j_{L_i-1}, 1}^{n, \sigma^*}\subseteq B_{j_{L_i-2}, 2}^{n, \sigma^*}$ under any $\sigma^{n,*}$ for sufficiently large $n$, and thus $\lim_{n\to\infty}|B_{j_{L_i-2}, 2}^{n, \sigma^*}|\geqslant\lim_{n\to\infty}|B_{j_{L_i-1}, 1}^{n, \sigma^*}|=\infty$ under any equilibrium $\sigma^{n,*}$.  Also we know that $\lim_{n\to\infty}|B_{j_{L_i-2},1}^{n,\sigma^*}|<\infty$.  Together with conditions (\ref{proofsi7}) and (\ref{proofsi8}), these facts validate conditions (\ref{proofsi3}) and (\ref{proofsi4}).  Hence we get that agent $j_{L_i-2}$ will not exit before she experiences her second communication step in any $\sigma^{n,*}$ provided that $n$ is sufficiently large.

The arguments above for $j_{L_i-2}$ can be extended successively to $j_{1}$.  Hence, under any equilibrium $\sigma^{n,*}$, no $j_{L_i-l}$ in the established path $\{j_{L_i-1}, j_{L_i-2}, ..., j_1, i\}$ will exit before she experiences $l$ communication steps under any equilibrium $\sigma^{n,*}$ provided that $n$ is sufficiently large.  A byproduct is that $\lim_{n\to\infty}|B_{j_{L_i-l}, l}^{n,\sigma^*}|=\infty$ under any $\sigma^{n,*}$, for $l\in\{1, 2, ..., L_i-1\}$.

\textsc{Step 3.} Finally, we argue that the socially informed agent $i$ will not exit before she experiences $L_i$ communication steps under any equilibrium $\sigma^{n,*}$ when $n$ is sufficiently large.  It requires that there exists $N\in\mathbb{N}$ such that for all social networks $G_n\in \{G_n\}_{n=1}^\infty$ with $n\geqslant N$, we have
\begin{equation} \label{proofsi9}
\psi-\frac{1}{\rho+\bar{\rho}|B_{i,L_i}^{n,\sigma^*}|}>0\,,
\end{equation}
and
\begin{equation} \label{proofsi10}
\bar{r}^{L_i}\left(\psi-\frac{1}{\rho+\bar{\rho}|B_{i,L_i}^{n,\sigma^*}|}\right) > \bar{r}^l\left(\psi-\frac{1}{\rho+\bar{\rho}|B_{i,l}^{n,\sigma^*}|}\right)\,,
\end{equation}
for all $l< L_i$.

Recall that we have already shown that agent $j_{L_i-l}$ in the constructed path will not exit before she experiences $L_i-l$ communication steps for $l\in\{1,2,...,L_i-1\}$, under any equilibrium $\sigma^{n,*}$ provided that $n$ is sufficiently large, which implies that $B_{j_{L_i-1},1}^{n,\sigma^*}\subseteq B_{j_2,L_i-2}^{n,\sigma^*}\subseteq ... \subseteq B_{j_1,L_i-1}^{n,\sigma^*}\subseteq B_{i,L_i}^{n,\sigma^*}$ under any $\sigma^{n,*}$ for sufficiently large $n$, and thus $\lim_{n\to\infty}|B_{i, L_i}^{n, \sigma^*}|\geqslant\lim_{n\to\infty}|B_{j_1, L_i-1}^{n, \sigma^*}|\geqslant ...\geqslant\lim_{n\to\infty}|B_{j_{L_i-2}, 2}^{n, \sigma^*}|\geqslant\lim_{n\to\infty}|B_{j_{L_i-1}, 1}^{n, \sigma^*}|=\infty$ under any $\sigma^{n,*}$.  Also, we have $B_{i,l}^{n,\sigma^*}\subseteq B_{i,l}^{n}$ and thus $|B_{i,l}^{n,\sigma^*}|\leqslant |B_{i,l}^{n}|$, under any $\sigma^{n,*}$ for $l\in\{1,2,...,L_i-1\}$, which implies the right hand sides of condition (\ref{incentiveinform}) are greater than or equal to th right hand sides of condition (\ref{proofsi10}), for $l\in\{1,2,...,L_i-1\}$.  Moreover, we know that $\lim_{n\to\infty}|B_{i,l}^{n,\sigma^*}|\leqslant \lim_{n\to\infty}|B_{i,l}^{n}|<\infty$ for $l\in\{1,2,...,L_i-1\}$ by the definition of $L_i$.  Together with conditions (\ref{positiveinform}) and (\ref{incentiveinform}) in Definition \ref{socialinformed}, these facts validate conditions (\ref{proofsi9}) and (\ref{proofsi10}).  Hence we get that the socially informed agent $i$ will not exit before she experiences $L_i$ communication steps and she can enjoy a positive payoff when she experiences $L_i$ communication steps,  under any $\sigma^{n,*}$ provided that $n$ is sufficiently large.  This further implies $k_{i}^{n,\sigma^*}\geqslant |B_{i, L_i}^{n,\sigma^*}|$ under any $\sigma^{n,*}$ with sufficiently large $n$, which finally leads to $\lim_{n\to\infty}|k_{i}^{n,\sigma^*}|\geqslant \lim_{n\to\infty}|B_{{i}, L_i}^n|=\lim_{n\to\infty}|B_{i,L_i}^n|=\infty$ under any $\sigma^{n,*}$ when $L_i\geqslant 2$. This concludes the proof.
\end{proof}

\end{document}